\documentclass[11pt]{article}
\usepackage[letterpaper,margin=1in]{geometry}
\usepackage[utf8]{inputenc}
\usepackage{microtype}
\usepackage{graphicx} \usepackage[dvipsnames]{xcolor}
\usepackage{mathtools,zref-savepos}
\usepackage{xfrac}
\usepackage{nicefrac}
\usepackage {amsthm}

\usepackage{url}
\DeclareUrlCommand\UScore{\urlstyle{rm}}
\newcommand{\doi}[1]{\href{https://doi.org/#1}{\UScore{#1}}}

\usepackage[colorlinks]{hyperref}
\hypersetup{citecolor=blue
}
\usepackage{bbold}
\usepackage{tikz}
\usetikzlibrary{shapes.geometric, positioning}
\usepackage{pstricks}
\usetikzlibrary{intersections, calc, angles}
\usepackage{subcaption}
\usepackage{natbib}
\usepackage{comment}
\usepackage{booktabs}

\makeatletter
\usetikzlibrary{decorations,backgrounds}
\def\pgfdecoratedcontourdistance{0pt}
\pgfset{
  decoration/contour distance/.code=\pgfmathsetlengthmacro\pgfdecoratedcontourdistance{#1}}
\pgfdeclaredecoration{contour lineto closed}{start}{\state{start}[
    next state=draw,
    width=0pt,
    persistent precomputation=\let\pgf@decorate@firstsegmentangle\pgfdecoratedangle]{\pgfpathmoveto{\pgfpointlineattime{.5}
      {\pgfqpoint{0pt}{\pgfdecoratedcontourdistance}}
      {\pgfqpoint{\pgfdecoratedinputsegmentlength}{\pgfdecoratedcontourdistance}}}}\state{draw}[next state=draw, width=\pgfdecoratedinputsegmentlength]{\ifpgf@decorate@is@closepath@ \pgfmathsetmacro\pgfdecoratedangletonextinputsegment{-\pgfdecoratedangle+\pgf@decorate@firstsegmentangle}\fi
    \pgfmathsetlengthmacro\pgf@decoration@contour@shorten{-\pgfdecoratedcontourdistance*cot(-\pgfdecoratedangletonextinputsegment/2+90)}\pgfpathlineto
      {\pgfpoint{\pgfdecoratedinputsegmentlength+\pgf@decoration@contour@shorten}
      {\pgfdecoratedcontourdistance}}\ifpgf@decorate@is@closepath@ \pgfpathclose
    \fi
  }\state{final}{}}
\makeatother
\tikzset{
  contour/.style={
    decoration={
      name=contour lineto closed,
      contour distance=#1
    },
    decorate}}

\usepackage{algpseudocodex}
\newcommand{\Bd}{\operatorname{Bd}}
\newcommand{\cost}{\operatorname{cost}}
\newcommand{\fcost}{\operatorname{cost}_{\mathsf{fac}}}
\newcommand{\ccost}{\operatorname{cost}_{\mathsf{conn}}}

\newcommand{\opt}{\mathsf{OPT}}
\newcommand{\alg}{\mathsf{ALG}}

\DeclareMathOperator{\conn}{conn}
\DeclareMathOperator{\fac}{fac}
\newcommand{\A}{\mathcal{A}}
\newcommand{\X}{X}
\renewcommand{\S}{\mathcal{S}}

\newcommand{\R}{\mathbb{R}}
\newcommand{\N}{\mathbb{N}}
\newcommand{\inst}{\mathcal{I}}
\newcommand{\1}{\mathbb{1}}
\newcommand{\adv}{\mathsf{ADV}}
\newcommand{\nmfl}{\text{NMFL}\xspace}
\newcommand{\setcov}{\text{SC}\xspace}
\newcommand{\nwst}{\text{NWST}\xspace}
\newcommand{\pcnwst}{\text{PC NWST}\xspace}
\newcommand{\nwsf}{\text{NWSF}\xspace}
\newcommand{\pcnwsf}{\text{PC NWSF}\xspace}

\newcommand{\LP}{\operatorname{LP}}
\newcommand{\Exp}{\mathbb{E}}

\makeatletter
\algnewcommand\algorithmicwhen{\textbf{when}}\algdef{SE}[WHEN]{When}{EndWhen}[1]{\algpx@startIndent\algpx@startCodeCommand\algorithmicwhen\ #1\ \algorithmicdo}{\algpx@endIndent\algpx@startCodeCommand\algorithmicend\ \algorithmicwhen}

\ifbool{algpx@noEnd}{\algtext*{EndWhen}\algtext*{EndCase}\apptocmd{\EndWhen}{\algpx@endIndent}{}{}}{}

\pretocmd{\When}{\algpx@endCodeCommand}{}{}

\ifbool{algpx@noEnd}{\pretocmd{\EndWhen}{\algpx@endCodeCommand[1]}{}{}}{\pretocmd{\EndWhen}{\algpx@endCodeCommand[0]}{}{}}\makeatother
\usepackage{algorithm}
\usepackage{amsmath}
\usepackage{amssymb}
\usepackage{amsthm}
\usepackage{tikz}
\newtheorem{claim}{Claim}
\newtheorem{theorem}{Theorem}
\newtheorem{proposition}{Proposition}
\newtheorem{lemma}{Lemma}
\theoremstyle{definition}

\usepackage{relsize}

\usepackage{xspace}
\usepackage{mathrsfs}
\newcommand*\samethanks[1][\value{footnote}]{\footnotemark[#1]}

\usepackage[capitalize]{cleveref}
\crefname{claim}{Claim}{Claims}
\Crefname{claim}{Claim}{Claims}

\title{Stronger adversaries grow cheaper forests:\\
online node-weighted Steiner problems}

\author{Sander Borst\thanks{This project has received funding from the European Research Council (ERC) under the European Union's Horizon 2020 research and innovation programme (grant agreement QIP--805241)}\\
	MPI Informatics\\
	{\small \texttt{sborst@mpi-inf.mpg.de}} \and 
    Marek Eli\'a\v{s}\thanks{
This study was funded by the European Union -  NextGenerationEU, in the framework of the FAIR - Future Artificial Intelligence Research project (FAIR PE00000013 – CUP B43C22000800006). The views and opinions expressed are solely those of the authors and do not necessarily reflect those of the European Union, nor can the European Union be held responsible for them.}\\Bocconi University\\
    {\small \texttt{marek.elias@unibocconi.it}} \and
    Moritz Venzin\samethanks\\Bocconi University\\
    {\small \texttt{moritz.venzin@unibocconi.it}}
 }

\date{}

\begin{document}

\maketitle

\begin{abstract}
\noindent
We propose a $O(\log k \log n)$-competitive randomized algorithm for online node-weighted
Steiner forest.
This is essentially optimal and significantly improves over the previous bound of
$O(\log^2 k \log n)$ by \citet{HajiaghayiLiaghatPanigrahi13}.
In fact, our result extends to the more general
prize-collecting setting, improving over previous works
by a poly-logarithmic factor.
Our key technical contribution is a randomized online algorithm
for set cover and non-metric facility location
in a new adversarial model which we call
\emph{semi-adaptive adversaries}.
As a by-product of our techniques, we
obtain the first deterministic $O(\log |C| \log |F|)$-competitive
algorithm for non-metric facility location.
\end{abstract}

\setcounter{page}{0}
\thispagestyle{empty}
\newpage
\section{Introduction}

The \emph{Steiner forest} problem is arguably one of the most central problems in network design. Given a graph $G = (V,E)$ and pairs of \emph{terminals} $T \subseteq V\times V$, the goal is to find a minimum weight subgraph connecting each pair. When the subgraph itself is required to be connected, this corresponds to the \emph{Steiner tree} problem. In this paper, we consider the \emph{online} version of the problem in the \emph{node-weighted} setting. Terminal pairs arrive one by one and we need to maintain a subgraph connecting the arrived pairs. The weight function is on nodes, i.e., to use an edge of the graph, we need to select its two incident nodes (vertices). Decisions are \emph{irrevocable}, meaning vertices cannot be un-selected, and the goal is to minimize the total cost of selected vertices.

The node-weighted setting is a significant generalisation of the edge-weighted setting\footnote{Any edge $e$ of weight $w(e)$ can be subdivided by a vertex of weight $w(e)$.} and unifies network design problems and covering problems such as \emph{set cover} and \emph{non-metric facility location}\footnote{See Subsection~\ref{subsec:online_problems} for a formal definition of these problems.}. This introduces considerable complications and the progress in the node-weighted setting has been relatively slow. The first breakthrough in this area was achieved by~\cite{NaorPanigrahiSingh11}. They showed how to extend a combinatorial structure called \emph{spider decomposition} that was first used in the \emph{offline} setting by~\cite{KleinRavi95} to be used in connection with an instance of randomized non-metric facility location. Their approach yields a \emph{poly-logarithmic} approximation to several online Steiner problems in the node-weighted setting. In particular, they achieve a randomized, $O(\log^2 k \log n)$-competitive online algorithm for the node-weighted Steiner tree problem. Here,~$k$ is the number of arriving terminal pairs and~$n$ is the number of vertices in the graph. For Steiner forest, they achieve a competitive ratio that is poly-logarithmic in~$n$, but the running time of their online algorithm was \emph{quasi-polynomial}, i.e.\,of the order $n^{\text{poly}\log n}$.

\citet{HajiaghayiLiaghatPanigrahi13} presented an elegant framework inspired by \emph{dual-fitting} called \emph{disk-paintings}, that yielded the first polynomial time online algorithm for node-weighted Steiner forest with a competitive ratio of $O(\log^2 k \log n)$. For graphs with an excluded minor of constant size, e.g.\,\emph{planar} graphs, their approach yields a $O(\log n)$-competitive algorithm. For the node-weighted Steiner tree problem,~\cite{HajiaghayiLiaghatPanigrahi14} then presented a $O(\log^2 n)$-competitive algorithm. In fact, their competitive ratio holds against the fractional optimum.
In addition, they provide a generic reduction to extend any algorithm
for the problems above to the \emph{prize-collecting} setting, only losing a $O(\log n)$-factor in the process. In this setting, together with the terminal pair, we are also given a penalty $p \in \R_{\geq 0} \cup \{+\infty\}$. We must then decide whether to augment the subgraph to connect the pair or whether to pay the penalty instead. These results are resumed in~Figure~\ref{fig:intro_past_best}.

On the other hand, the lower bound of (online) node-weighted Steiner problems comes from (online) set cover: an instance of set cover with $k$ elements and $n$ sets already reduces to node-weighted Steiner tree, see Figure~\ref{fig:SetCover_to_NWSF} for an illustration.
The lower bound of $\smash{\Omega(\tfrac{\log k \log n}{\log\log k + \log\log n})}$
by \citet{Alon_et_al_SetCover}, which holds for $k=\Omega(\log n)$,
therefore applies. Assuming $\text{P} \neq \text{NP}$, there is a stronger lower bound of $\Omega(\log k \log n)$ by~\cite{Korman04}, which holds for polynomial-time online algorithms whenever $k$ and $n$ are polynomially related.

In this work, we almost close this (poly-)logarithmic gap:
we propose a randomized $O(\log k\log n)$-competitive online algorithm for the (prize-collecting) node-weighted Steiner forest which runs in polynomial time.
We also obtain a $O(\log^2 n)$-competitive deterministic algorithm.
This improves and subsumes all previous bounds for these problems (Figure~\ref{fig:intro_past_best}).

\begin{figure}
    \centering
    \begin{tabular}{lcccc}
& \nwst
        & \nwsf
        & PC \nwst
        & PC \nwsf
        \\
        \midrule
         Previous works:
         & $\begin{matrix} 
            O(\log^2 n)\\
            O(\log^2 k \log n)
            \end{matrix}$
         & $O(\log^2 k\log n)$
         & $\begin{matrix}
            O(\log^3 n)\\
            O(\log^2 k \log^2 n)
            \end{matrix}$
         & $O(\log^2 k\log^2 n)$
\end{tabular}
    \caption{The best competitive ratios achieved by previous works for the online node-weighted Steiner tree (\nwst),
    node-weighted Steiner forest (\nwsf), and their prize collecting variants (\text{PC}). All these algorithms are polynomial time and randomized.
    }
    \label{fig:intro_past_best}
\end{figure}

\subsection{Our results}\label{subsec:results}

Our main result is a new randomized online algorithm for node-weighted prize-collecting Steiner forest.
It improves over the competitive ratio of
the previous online algorithms
by \citet{HajiaghayiLiaghatPanigrahi13}
and
\citet{HajiaghayiLiaghatPanigrahi14}
for node-weighted Steiner tree (\nwst), node-weighted Steiner forest (\nwsf), and
their prize-collecting variants (\text{PC}), see
Figure~\ref{fig:intro_past_best} for the summary of the previous bounds.

\begin{theorem}
\label{thm:intro_nwsf_rand}
There is a randomized online algorithm for prize-collecting
node-weighted Steiner forest
which is $O(\log k \log n)$-competitive
against an oblivious adversary
on a graph with $n$ vertices,
where $k$ denotes the number of terminal pairs received online.
\end{theorem}

Note that the setting with oblivious adversaries is standard
for the analysis of randomized online algorithms and was also used by
previous works \citep{NaorPanigrahiSingh11,HajiaghayiLiaghatPanigrahi13,HajiaghayiLiaghatPanigrahi14}.

We also provide a deterministic online algorithm for \pcnwsf. Such an algorithm is analyzed against an adaptive adversary. A set $T\subseteq V\times V$ of possible terminals pairs is given to the online algorithms, and the adversary chooses $k$ of them adaptively, based on the previous actions performed by the online algorithm. By relation to online set cover, the competitive ratio
then depends on $\log |T|$ instead of $\log k$,
see~\citep{bienkowski_nearly_2021}. The competitive ratio of our algorithm for \pcnwsf is $O(\log|T|\log n)$, improving over the previous best, $O(\log |T| \log^2 n)$-competitive ratio by \citet{bienkowski_nearly_2021} in the non-prize-collecting, deterministic
setting.

\begin{theorem}
\label{thm:intro_nwsf_det}
Let $G =(V,E)$ be a graph on $n$ vertices and $T \subseteq V\times V$
be a set of possible terminal pairs.
There is a deterministic
online algorithm for
prize-collecting node-weighted Steiner forest
which knows $T$ beforehand and achieves competitive ratio
$O(\log |T|\, \log n)$ on any online input sequence
of terminal pairs chosen from $T$.
\end{theorem}

Since $T \subseteq V\times V$, our deterministic algorithm is always $O(\log^2 n)$-competitive. Hence we recover the competitive ratio
of the randomized algorithm for \nwst by
\citet{HajiaghayiLiaghatPanigrahi14} using a deterministic algorithm
and improve over their bounds for the prize-collecting variants,
see Figure~\ref{fig:intro_past_best}.

Similarly to \citep{HajiaghayiLiaghatPanigrahi13} and \citep{HajiaghayiLiaghatPanigrahi14},
our algorithm solves an auxiliary instance of non-metric facility location (\nmfl).
Our auxiliary instance is different to theirs and
is generated adaptively, based on the previous
steps of the algorithm.
Therefore, we cannot solve this instance with existing randomized
online algorithms \citep{alon_general_2006} which work against
oblivious (non-adaptive) adversaries.
However, the adaptivity is constrained in
a certain way which still allows for algorithms with a good
performance.
This motivates the study of \nmfl and the related set cover problem (\setcov) in the setting
with {\em semi-adaptive} adversaries -- a new setting
which captures the difficulty of solving the auxiliary
\nmfl instance in a refined way and might be of independent interest.

\paragraph{\nmfl and \setcov with semi-adaptive adversaries.}
We study the online non-metric facility location problem (\nmfl) and the
online set cover problem (\setcov) in a new adversarial setting. This setting interpolates between the the classical setting of deterministic algorithms for \setcov and \nmfl
and the setting of randomized algorithms against oblivious adversaries.
Here, we force the adversary
to fix a super-instance~$I$ beforehand
and restrict its adaptiveness
to selecting parts of $I$ based on
the (random) decisions of the algorithm.
We call such an adversary \emph{semi-adaptive}.
In particular, for a super-instance $I=(C,F)$ composed of clients $C$ and facilities $F$ fixed by a semi-adaptive adversary, the adversary can adaptively select a subset of the clients $C' \subseteq C$ to be presented to the online algorithm. 
Similarly to the performance bounds for deterministic algorithms,
the competitive ratio of our online randomized algorithm then depends on $|C|$ (and $|F|$), not on the size of $C'$. 
However, contrary to the deterministic setting, our online algorithm does not require the knowledge
of $C$ and $F$ beforehand.

\begin{theorem}
\label{thm:intro_nmfl}
There is an online randomized algorithm for non-metric facility location
against a semi-adaptive adversary
with the following guarantee.
Denoting 
$\alg$ and $\opt$ the cost of the online
algorithm and offline optimum respectively, we have
\[
\Exp[\alg] \leq O(\log |C| \log |F|)\cdot \Exp[\opt],
\]
where $C$ is the set of clients and $F$ is the set of
facilities in the super-instance $I$ fixed by the adversary beforehand,
and both expectations are computed
over the random choices of the online algorithm. Additionally, denoting by $\ccost(\alg)$ the connection cost incurred by our online algorithm, it holds that 
\[
\Exp[\ccost(\alg)] \leq O(\log |F|)\cdot \Exp[\opt].
\]
\end{theorem}

This more refined bound on the cost incurred on connections and facilities is crucial for our algorithm for \nwsf.
Theorem~\ref{thm:intro_nmfl} applies to \setcov in the semi-adaptive setting as a special case of \nmfl:
$C$ and $F$ correspond to elements and sets
respectively in the super-instance of \setcov fixed by the adversary. The connection costs are either $0$ or $+\infty$ and can be disregarded. Interestingly, the main difficulty of the proof of Theorem~\ref{thm:intro_nmfl} lies in constructing a randomized rounding scheme for \setcov which works well against semi-adaptive adversaries. In turn, we deterministically reduce the online rounding of the natural LP relaxation of \nmfl \citet{alon_general_2006} to \setcov rounding (against a semi-adaptive adversary).

Plugging in the deterministic rounding scheme for \setcov by~\citet{Alon_et_al_SetCover} instead, we obtain a deterministic $O(\log |C|\log|F|)$-competitive algorithm for \nmfl. This improves over the reduction from \nmfl to \setcov due to \citet{KolenTamir1990} that yields a $O((\log|C|+\log|F|)\cdot(\log|C|+\log\log|F|))$-competitive deterministic algorithm, as well as the recent, $O(\log|F|\cdot(\log|C|+\log\log|F|)$-competitive deterministic algorithm from~\citet{bienkowski_nearly_2021}\footnote{Their algorithm has the same refined guarantee on cost incurred on connections, see Lemma~$10$ in~\citet{bienkowski_nearly_2021}. Using this as a subroutine for \nmfl in our algorithm then yields a deterministic $O(\log n \cdot (\log |T| + \log\log n))$-competitive algorithm for \pcnwsf.}. The latter algorithm is quite involved: they pass through a new LP-relaxation for \nmfl for which they design a new online rounding scheme which can handle both connection and facility opening costs. While conceptually similar, our approach is much simpler as most of the complexity is encapsulated in the set cover algorithm that we use in a black-box way. The improved competitive ratio is a result of the observation that all connection costs smaller than $\smash{\frac{\opt}{|C|}}$ can be assumed to be~$0$, while they only do this for all costs smaller than~$\frac{\opt}{|C|\cdot|F|}$

\begin{theorem}
There is a deterministic $O(\log |C| \log |F|)$-competitive algorithm for
non-metric facility location, where $C$ is the set of clients and $F$ is the set of facilities in the super-instance $I$ fixed by the adversary beforehand and known to the algorithm. Denoting by $\ccost(\alg)$ the connection costs incurred by this deterministic algorithm, it additionally holds that $\ccost(\alg) \leq O(\log|F|)\cdot \opt$, where $\opt$ is the optimum value to the instance selected by the adversary.
\end{theorem}

\paragraph{Organization of the paper.}
In Sections~\ref{subsec:related_works} and~\ref{subsec:techniques} we give an overview on the related work, and our techniques, respectively. The formal definitions of the online problems considered in this paper as well the adversarial models considered can be found in Sections~\ref{subsec:online_problems} and~\ref{subsec:comp_ratio_and_adversaries}.
Section~\ref{sec:nwsf} contains our algorithm for \nwsf.
In Sections~\ref{subsec:rounding_sc} and \ref{subsec:online_nmfl}, we present our results
for \setcov and \nmfl respectively.

\subsection{Related works}\label{subsec:related_works}

Since its inception in the $1800$s, the Steiner tree problem and its generalisations have played a prominent role in the development of geometry and combinatorics, and, in more recent years, in theoretical computer science and practical areas such as the design of infrastructure and communications networks; see~\cite{SteinerTreeHistory, Ljubic2021} for a detailed overview. Being one of Karp's original $21$ NP-hard problems,~\cite{Karp1972}, Steiner problems are ubiquitous in the design and analysis of (approximation) algorithms, specifically in the area of network design, and have sparked significant algorithmic advances. Most relevant to our setting are works on the node-weighted setting in the setting of offline algorithms, and works online algorithms in the edge-weighted setting.

The study on node-weighted Steiner problems in the offline setting was initiated by~\cite{KleinRavi95}, giving a $O(\ln k)$-approximate algorithm to the node-weighted Steiner forest problem. We note that their approach inspired the first online algorithm for \nwst by~\cite{NaorPanigrahiSingh11}. Subsequently, this was extended to the prize-collecting setting by~\cite{MossRabani2007}. Their original proof contained a flaw, but a new algorithm with an approximation ratio of $O(\log k)$ (resp. $O(\log n)$ for \pcnwst) was given by~\cite{BateniHajiaghayiLiaghat18,chekuri_prize-collecting_2012,KoenemannSadeghianSanita13}.

In the edge-weighted and online setting, there has been considerable work on Steiner and network optimization problems. This line of research was initiated by~\cite{ImaseWaxman91} who showed that the natural greedy algorithm for online Steiner tree is $O(\log k)$-competitive, and that this is optimal. This result was extended to the case of online Steiner forest through the \emph{augmented greedy} algorithm by~\cite{BermanCoulston97}. Interestingly, the (natural) greedy algorithm is only known to be $O(\log^2 k)$-competitive for Steiner forest,~\cite{AwerbuchAzarBartal96}, though there is recent evidence its performance should be much better,~\cite{BamasDrygalaMaggiori22}. Around the same time, there was also significant progress on online network optimization problems. This was initiated by the work of~\cite{Alon_et_al_SetCover}, who provided a $O(\log |X| \log |\S|)$-competitive algorithm for online set cover. This result was extended to a wider range of problems such as non-metric facility location, and the technique was later generalized to the \emph{online primal-dual method},~\cite{alon_general_2006, BuchbinderNaor09}. We also mention several extensions of these problems that were considered recently, such as online constrained forest problems, generalising over the Steiner forest problem,~\cite{QianWilliamson11, QianUmbohWilliamson18}, \emph{online survivable network design}, extending to higher connectivity requirements,~\cite{GuptaKrishnaswamyRavi12}, as well as frameworks for Steiner problems in the \emph{deadline and delay} setting~\cite{AzarTouitou20}, a setting that is similar though incomparable to the online setting. Finally, we note that the \pcnwsf problem was recently also considered by~\cite{Markarian_2018}. She provides a \emph{fractional} $O(\log n)$-competitive algorithm and claims this monotonically increasing solution can be rounded online incurring a multiplicative loss of $O(\log k)$. The latter part is incorrect, we show this in~\cref{app:flaw}.

\subsection{Our techniques}\label{subsec:techniques}

\paragraph{NWSF with a single NMFL instance.}
From a high-level perspective, our algorithm for \nwsf is similar to
\citep{HajiaghayiLiaghatPanigrahi13}:
we combine an algorithm for \nmfl and the classical augmented greedy algorithm
designed for online (edge-weighted) Steiner forest.
While the algorithm of \citet{HajiaghayiLiaghatPanigrahi13}
generates $\log k$ separate instances of \nmfl and serves them by $\log k$
copies of an algorithm for \nmfl run in parallel,
we only use a single \nmfl instance.
This instance is generated
by our algorithm for \nwsf and contains one of $\log k$
possible copies of each client (each with different connection costs).

Our algorithm for \nwsf works against an oblivious adversary. Such an adversary can be thought of having to fix the whole instance of \nwsf
beforehand, see Section~\ref{subsec:comp_ratio_and_adversaries}
for details on various adversarial models. However, which one of the $\log k$ clients arrives to the \nmfl instance depends on previous decisions of the \nmfl algorithm. Therefore, the input sequence for the \nmfl instance is not determined beforehand, hence, we cannot assume that the algorithm faces an oblivious adversary. However, we argue that this situation can be remedied by considering a semi-adaptive adversary:
as the adversary has to fix the input instance of \nwsf in advance, this determines a super-instance $I'$ of \nmfl with $n$ facilities
and $\log k\cdot k$ potential clients. The decision which of these potential clients are presented to the \nmfl algorithm is then deferred to an adversary, allowing us to use our results for \nmfl in the setting with semi-adaptive adversaries.

To extend our algorithm to prize-collecting setting,
we suitably modify the auxiliary \nmfl instance, accommodating the
penalty costs.
This approach is different from \citet{HajiaghayiLiaghatPanigrahi14}
who provided a black-box reduction
from problems belonging to a certain class of network design
problems containing \nwsf
to their prize-collecting variants.
This reduction
loses a factor $O(\log n)$ on top of the competitive ratio
of the algorithm for the non-prize-collecting variant.

\paragraph{Set cover with semi-adaptive adversaries.}
Our rounding algorithm for set cover is very similar
to the algorithm of \cite{Alon_et_al_SetCover}:
we choose $\log |X|$ rounding thresholds for each set,
where $|X|$ is the number
of elements in the super-instance $I' = (X, \S)$ fixed
by the adversary beforehand. Although our algorithm
is not instance-aware (does not know $I'$), $\log |X|$ can be guessed online.
We purchase a set whenever its fractional value exceeds
any of the rounding thresholds.
Here we emphasize the difference to the oblivious setting: there we only need to choose $\log k$ rounding thresholds, where $k$ is the number of elements that will actually arrive.

While this online rounding scheme can be considered classic,
proving its competitive ratio in the semi-adaptive setting is more challenging.
We model the evolution of the solutions as a stochastic process, and use martingales to bound the probability that a single element is uncovered on arrival. Then we take a union bound over all potentially arriving elements from the superinstance. Hence, the size of the superinstance influences the number of terms in the union bound. This is the reason that the competitive ratio of the rounding scheme depends on this size.

\paragraph{Connection-aware set cover rounding for NMFL}
As pointed out by \citet{bienkowski_nearly_2021},
a direct rounding of the fractional values of the facilities
in a set-cover manner ignores the connection costs,
and leads to a very bad performance.
To handle this issue, we use a super-instance of the set cover problem. Each facility corresponds to a set, but for each client $c$, the instance contains $\log |C|$ distinct elements. The $j^{\text{th}}$ such element is only covered by sets for which the connection cost of client $c$ to the respective facility is at most $\tfrac{\opt}{2^j}$. Upon the online arrival of a client, based on the fractional solution to the non-metric facility location instance, we select one of the $\log|C|$ elements corresponding to the client to be released to an instance of set cover, and open the facilities corresponding to sets picked. Finally, we greedily connect the client to an opened facility.

\section{Preliminaries}\label{sec:preliminaries}
This section is divided in three parts. In Subsection~\ref{subsec:graph_theory} we start by recalling the relevant notions and definitions from graph theory. In Subsection~\ref{subsec:online_problems} we formally state the online (prize-collecting) node-weighted Steiner forest problem as well as the related online non-metric facility location problem and the online set cover problem. Finally, in subsection~\ref{subsec:comp_ratio_and_adversaries}, we give an overview on the types of online adversaries and the respective measure of competitiveness that we consider.

\subsection{Notions from Graph Theory}\label{subsec:graph_theory}
Throughout this paper, we consider undirected graphs $G = (V, E)$ containing $n$ vertices $(|V| =  n)$ along with a node-weight function $w: V \rightarrow \mathbb{R}_{\geq 0}$. For a subset of vertices $S\subseteq V$, we define $G[S]$ to be the subgraph induced by $S$, that is, the subgraph with vertex set $S$ and whose edge set consists of all edges with both endpoints in $S$. Its weight is denoted by $w(S) := \sum_{v\in S}w(v)$. For any pair of vertices $u,v \in V$, we define $d_G(u,v)$ to be the weight of the cheapest path between $u$ and $v$ with respect to the node-weight function $w$ (excluding the weights of $u$ and $v$). 
For a set $S$ of nodes we also define $d_{G/S}(w,v)$ to the be the weight of the cheapest path between $w$ and $v$, not counting the weight of nodes of $S\cup \{w,v\}$. 
Analogously, given $S \subseteq V$, we denote by $G/S$ the graph $G$ where the node-weights of all vertices in $S$ have been set to $0$.
Finally, we define the open ball of radius $r$ centered at $u$ as
$B(u,r) := \{ v \in V \mid d_G(u,v) < r \}$
and the boundary of $B(u,r)$ as
\[
\Bd(u, r) := \{v\in V \mid d_G(u,v) < r \leq d_G(u,v) +w_v \}.
\]

\subsection{Online Problems (\pcnwsf, \nmfl, \setcov)}\label{subsec:online_problems}

In \emph{online prize-collecting node-weighted Steiner forest} (\pcnwsf),
we are given
a graph $G=(V,E)$ and a cost function $w\in \R_{\geq 0}^V$ beforehand.
At each time $i=1, \dotsc, k$,
we receive a pair of terminals $(s_i, t_i)\in V\times V$
and a penalty cost $p_i \in \R_{\geq 0} \cup\{+\infty\}$.
Our task is to decide whether to augment a set $S\subseteq V$
so that the induced subgraph $G[S]$ connects $s_i$ and $t_i$
or to pay penalty $p_i$.
The objective is to minimize the sum of the paid penalties, plus the
weight of the vertices in $S$, i.e.\,$w(S) = \sum_{v\in S} w_v$.
This generalizes the \emph{prize-collecting node-weighted Steiner tree problem} (\pcnwst), where we need to maintain a connected subgraph (i.e. a tree) connecting terminals $s_1, s_2, \ldots, s_k$ (or pay the respective penalty $p_i$). Non-prize-collecting variants correspond to the special case where
$p_i = +\infty$ for each $i=1, \dotsc, k$.

An instance of \emph{online non-metric facility location} (\nmfl)
is given by a \emph{client-facility graph} $G := (F, C, \cost)$.
$G$ is a complete bipartite graph between clients $C$ and facilities $F$. Each facility $f$ has an opening cost $\fcost: F \rightarrow \R_{\geq 0}$, and each edge $(c,f)$ has a connection cost $\ccost: C\times F \rightarrow \R_{\geq 0} \cup \{+\infty\}$. 
At each time step $i = 1, \ldots, k$, we are given a client $c_i$, along with all facilities (and the corresponding facility and connection costs) for which $\ccost(c_i, f) < +\infty$. An online algorithm has to immediately connect it to some facility $f\in F$,
paying $\fcost(f) + \ccost(c_i,f)$ if the facility is not yet open, or $\ccost(c_i,f)$ otherwise.

An instance of \emph{online set cover} (\setcov) is given by a set system $(X, \S)$ and a cost function $c: \S \rightarrow \R_{\geq 0}$ on the sets. At each time step $i = 1, \ldots, k$, we are given an element $e\in \X$, along with all sets $S \in \S$ containing $e$, as well as their respective cost. If the element is not yet covered, an online algorithm needs to select some set containing $e$ and pay the cost of the respective set. Note that this corresponds to an instance of \nmfl with $\ccost \in \{0, +\infty\}$.
Set cover can be seen as a special case of \nwsf, see Figure~\ref{fig:SetCover_to_NWSF} for illustration.

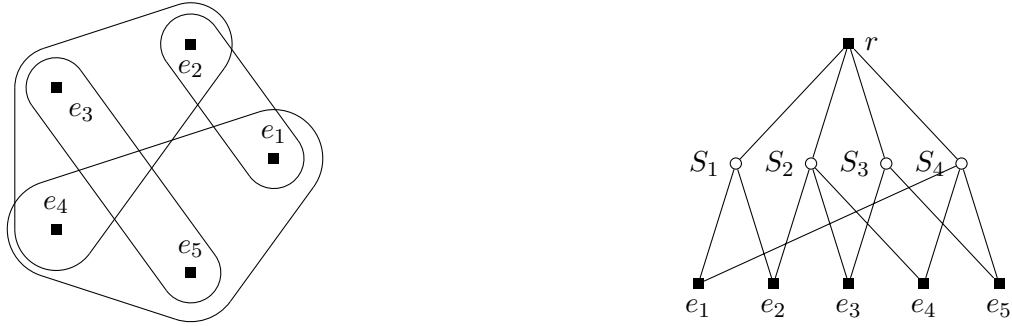
\begin{figure}[htbp]
    \centering
\begin{subfigure}[b]{0.45\textwidth}
        \centering
        \begin{tikzpicture}
 
\def\radius{.65cm}
\node[label = $e_1$, rectangle, inner sep = 2pt, fill] (e1) at (0:1.6) {};
\coordinate[label = -90:$e_2$, rectangle, inner sep = 2pt, fill] (e2) at (72:1.6) {};
\coordinate[label = -70:$e_3$, rectangle, inner sep = 2pt, fill] (e3) at (144:1.6) {};
\coordinate[label = $e_4$, rectangle, inner sep = 2pt, fill] (e4) at (216:1.6) {};
\coordinate[label = $e_5$, rectangle, inner sep = 2pt, fill] (e5) at (288:1.6) {};

\coordinate (e11) at ($(e1)!\radius!-90:(e4)$);
\coordinate (e41) at ($(e4)!\radius!+90:(e1)$);
\coordinate (e42) at ($(e4)!\radius!-90:(e5)$);
\coordinate (e51) at ($(e5)!\radius!+90:(e4)$);
\coordinate (e52) at ($(e5)!\radius!-90:(e1)$);
\coordinate (e12) at ($(e1)!\radius!+90:(e5)$);

\draw (e11)--(e41) (e42)--(e51) (e52) -- (e12);
\pic [draw, angle radius=\radius] {angle=e12--e1--e11};
\pic [draw, angle radius=\radius] {angle=e41--e4--e42};
\pic [draw, angle radius=\radius] {angle=e51--e5--e52};

\def\radius{.55cm}
\coordinate (e21) at ($(e2)!\radius!-90:(e3)$);
\coordinate (e31) at ($(e3)!\radius!90:(e2)$);
\coordinate (e32) at ($(e3)!\radius!-90:(e4)$);
\coordinate (e41) at ($(e4)!\radius!90:(e3)$);
\coordinate (e42) at ($(e4)!\radius!-90:(e2)$);
\coordinate (e22) at ($(e2)!\radius!90:(e4)$);

\draw (e21) -- (e31) (e32) -- (e41) (e42) -- (e22);
\pic [draw, angle radius=\radius] {angle=e22--e2--e21};
\pic [draw, angle radius=\radius] {angle=e31--e3--e32};
\pic [draw, angle radius=\radius] {angle=e41--e4--e42};

\def\radius{.4cm}
\coordinate (e31) at ($(e3)!\radius!90:(e5)$);
\coordinate (e32) at ($(e3)!\radius!-90:(e5)$);
\coordinate (e51) at ($(e5)!\radius!90:(e3)$);
\coordinate (e52) at ($(e5)!\radius!-90:(e3)$);

\draw (e31)--(e52) (e32)--(e51);
\pic [draw, angle radius=\radius] {angle=e31--e3--e32};
\pic [draw, angle radius=\radius] {angle=e51--e5--e52};

\def\radius{.4cm}
\coordinate (e21) at ($(e2)!\radius!90:(e1)$);
\coordinate (e22) at ($(e2)!\radius!-90:(e1)$);
\coordinate (e11) at ($(e1)!\radius!90:(e2)$);
\coordinate (e12) at ($(e1)!\radius!-90:(e2)$);

\draw (e21) -- (e12) (e22) -- (e11);
\pic [draw, angle radius=\radius] {angle=e21--e2--e22};
\pic [draw, angle radius=\radius] {angle=e11--e1--e12};

\end{tikzpicture}
        
    \end{subfigure}
    \hfill
\begin{subfigure}[b]{0.45\textwidth}
        \centering
            \begin{tikzpicture}
    \node[label = -90:$e_1$, rectangle, inner sep = 2pt, fill] (e1) at (-2, -1.6) {};
    \coordinate[label = -90:$e_2$, rectangle, inner sep = 2pt, fill] (e2) at (-1, -1.6) {};
    \coordinate[label = -90:$e_3$, rectangle, inner sep = 2pt, fill] (e3) at (-0, -1.6) {};
    \coordinate[label = -90:$e_4$, rectangle, inner sep = 2pt, fill] (e4) at (1, -1.6) {};
    \coordinate[label = -90:$e_5$, rectangle, inner sep = 2pt, fill] (e5) at (2, -1.6) {};

    \coordinate[label = -180:$S_1$, circle, inner sep = 1.5pt, draw] (S1) at (-1.5, 0) {};
    \coordinate[label = -180:$S_2$, circle, inner sep = 1.5pt, draw] (S2) at (-0.5, 0) {};
    \coordinate[label = -180:$S_3$, circle, inner sep = 1.5pt, draw] (S3) at (0.5, 0) {};
    \coordinate[label = -180:$S_4$, circle, inner sep = 1.5pt, draw] (S4) at (1.5, 0) {};

    \coordinate[label = 0:$r$, rectangle, inner sep = 2pt, fill] (r) at (0, 1.6) {};
    \draw (r)--(S1) (r)--(S2) (r)--(S3) (r)--(S4);
    \draw (e1)--(S1) (e2)--(S1) (e2) -- (S2) (e3) -- (S2) (e4) -- (S2) (e3) -- (S3) (e5) -- (S3) (e4) -- (S4) (e5) -- (S4) (e1) -- (S4);

    \end{tikzpicture}
    \end{subfigure}
    \caption{Reducing an instance of set cover to node-weighted Steiner forest (tree): for each element $e \in \X := \{e_1, \ldots, e_k\}$ in the ground set $\X$ is connected to set $S \in \mathcal{S} := \{S_1, \ldots, S_n\}$ if and only if $e\in S$. Every set is connected to some distinguished vertex $r$. The demand pairs are $\{(e_1, r), (e_2, r), \ldots, (e_k, r)\}$.}
    \label{fig:SetCover_to_NWSF}
\end{figure}

\subsection{Competitive ratio and different types of adversaries}\label{subsec:comp_ratio_and_adversaries}
An online problem can be seen as a $2$-player game between an algorithm and an adversary.
The adversary reveals the problem gradually and the algorithm needs to come up with a solution. The performance of the algorithm is then measured using the notion of competitive ratio, the definition of which depends on the type of adversary the algorithm faces. We give a brief overview on the various adversaries considered in this paper, we refer to~\cite[Chapter~$7$]{BorodinYaniv98} for a thorough, textbook exposition.

\paragraph{Oblivious adversary.}
This is the most common adversarial model for randomized algorithms.
An oblivious adversary knows the description of the algorithm
and has to commit to the whole problem instance before the algorithm starts. The instance is then revealed to the algorithm in an online fashion.
An algorithm $\alg$ is said to  be $\gamma$-competitive against an oblivious adversary if we have
\begin{align*}
  \Exp[\alg(I)] \leq \gamma\cdot \opt(I)
\end{align*}
for every input instance $I$,
where $\alg(I)$ denotes
the cost of the algorithm's solution to $I$,
$\opt(I)$ is the cost of the optimal solution to $I$ computed offline,
and the expectation is over the randomness of $\alg$.

\paragraph{Adaptive adversary.}
An adaptive adversary does not have to fix an instance in advance.
Instead, at each time step,
it chooses the next part of the input instance
$I_{\adv}$ based on (random) actions of the
algorithm so far.
Algorithm $\alg$ is said to be $\gamma$-competitive against adaptive adversaries
if for any adaptive adversary $\adv$ we have:
\begin{align*}
\Exp[\alg(I_{\adv})] \leq \gamma \cdot \Exp[\opt(I_{\adv})],
\end{align*}
where $\alg(I_{\adv})$ denotes
the cost of the algorithm's solution to $I_{\adv}$ and
$\opt(I_{\adv})$ is the cost of the optimal solution to $I_{\adv}$ computed offline.
Note that $\opt(I_{\adv})$ here also depends on the 
randomness of $\alg$, as the actions of the adversary are adaptive to those of the algorithm.
For many problems, such as online set cover, adaptive adversaries are strictly more powerful than their oblivious counterparts.

\paragraph{Semi-adaptive adversary}
In this paper we will introduce a new type of adversary, the semi-adaptive adversary. A semi-adaptive adversary has to commit to a super-instance $I$ of the problem instance beforehand, without revealing it to the algorithm. The actual problem instance can then be chosen adaptively based on the actions of the algorithm, but must be a subinstance of $I$.
The competitive ratio is defined as in the adaptive setting.

\paragraph{Competitiveness of deterministic algorithms.}
In the case of deterministic algorithms, oblivious and
adaptive adversaries coincide, since the adversary knows
the algorithm and can simulate its (deterministic) steps beforehand.
A deterministic algorithm $\alg$ is said to be $\gamma$-competitive,
if
\[ \alg(I) \leq \gamma \cdot \opt(I) \]
holds for every input instance $I$,
where $\alg(I)$ denotes
the cost of the algorithm's solution to $I$, and
$\opt(I)$ is the cost of the optimal solution to $I$ computed offline.

In case of \setcov, \nmfl, and \nwsf, deterministic algorithms are studied
in the setting where a super-instance of $I$ is known to the
algorithm beforehand\footnote{Already for online set cover, if the set system $(X,\S)$ is not known in advance, the competitive ratio of any deterministic algorithm is at least $|X|$, see~\cite{Alon_et_al_SetCover}.}. Here, $\gamma$
usually depends on the parameters of the super-instance instead of
the parameters of $I$.

\paragraph{Refined competitive ratio for \nmfl.}
We use a refined definition
of $(\alpha,\beta)$-competitiveness for \nmfl. Given an algorithm $\alg$ for \nmfl, we split the cost incurred by $\alg$ into the cost incurred on connection costs and facility costs, $\ccost(\alg)$ and $\fcost(\alg)$ respectively. We say $\alg$ is $(\alpha, \beta)$-competitive, if $\fcost(\alg)$ is $\alpha$-competitive and $\ccost(\alg)$ is $\beta$-competitive with respect to an optimum offline solution (which has to pay both facility and connection costs). Since $\cost(\alg) = \fcost(\alg) + \ccost(\alg)$, any $(\alpha, \beta)$-competitive algorithm is also $(\alpha+\beta)$-competitive.

\section{An algorithm for node weighted Steiner forest}
\label{sec:nwsf}
The main result of this section is our $O(\log k\log n)$-competitive algorithm for online (prize-collecting) node-weighted Steiner forest. In Subsection~\ref{subsec:overview} we give an overview of our approach.
In Section~\ref{subsec:alg_nwsf} we describe our algorithm for \nwsf (Algorithm~\ref{alg:nwsf}).
Finally, in Section~\ref{subsec:alg_pcnwsf}, we extend it to the prize-collecting setting of \nwsf.

\subsection{Overview}\label{subsec:overview}
We begin by describing our approach. Without loss of generality, we may assume that the online algorithm knows $k$, the number of eventually arriving terminal pairs, and that all arriving pairs have distance between $1$ and $k$, e.g. $d_{G/S}(s,t) \in [1, k]$. This is a standard assumption, we sketch the argument in the Appendix, see~\ref{subsec:guesskandopt}. 

The description of our algorithm is very similar to
\citep{HajiaghayiLiaghatPanigrahi13}.
From a high-level perspective, it is an extension of the Augmented Greedy algorithm by \citet{BermanCoulston97} for edge-weighted Steiner forest to the node-weighted setting through the elegant use of non-metric facility location.

\begin{figure}[h]
    \centering
\begin{tikzpicture}
\draw [gray, fill=gray!12, color=gray!12, dotted] plot [smooth cycle] coordinates {(0,2) (1.4,1.4) (2,0) (1,-1.5) (0, -0.8) (-3,0)};
\draw[color = gray!100] (0.6,0.7) circle (1cm){};
\coordinate[rectangle, fill, draw, inner sep = 1.5pt] (t1) at (0.6, 0.7) {};
\draw[color = gray!100] (-1.5,0.3) circle (0.7cm){};
\coordinate[rectangle, fill, draw, inner sep = 1.5pt] (t2) at (-1.5,0.3) {};
\draw[color = gray!100] (0.9,-0.9) circle (0.4cm){};
\coordinate[rectangle, fill, draw, inner sep = 1.5pt] (t3) at (0.9, -0.9) {};
\path (t1) ++(45:1cm) coordinate (t1fboundary);
\coordinate[circle, draw, inner sep = 1.5pt, fill=gray!70] (t1fboundary) at (t1fboundary) {};
\path (t1) ++(35:0.5cm) coordinate (t1finterior);
\coordinate[circle, draw, inner sep = 1pt, fill=gray!70] (t1finterior) at (t1finterior) {};

\path (t1) ++(-80:0.6cm) coordinate (t1greedy1);
\coordinate[circle, draw, inner sep = 1pt] (t1greedy1) at (t1greedy1) {};
\path (t1) ++(-95:1cm) coordinate (t1greedy2);
\coordinate[circle, draw, inner sep = 1pt] (t1greedy2) at (t1greedy2) {};

\path (t2) ++(160:0.7cm) coordinate (t2fboundary);
\coordinate[circle, draw, inner sep = 1.5pt, fill=gray!70] (t2fboundary) at (t2fboundary) {};
\path (t2) ++(220:0.45cm) coordinate (t2greedy1);
\coordinate[circle, draw, inner sep = 1pt] (t2greedy1) at (t2greedy1) {};
\path (t2) ++(200:0.7cm) coordinate (t2greedy2);
\coordinate[circle, draw, inner sep = 1pt] (t2greedy2) at (t2greedy2) {};
\path (t2) ++(-20:0.7cm) coordinate (t2greedy3);
\coordinate[circle, draw, inner sep = 1pt] (t2greedy3) at (t2greedy3) {};

\path (t3) ++(30:0.4cm) coordinate (t3fboundary);
\coordinate[circle, draw, inner sep = 1.5pt, fill=gray!70] (t3fboundary) at (t3fboundary) {};
\path (t3) ++(110:0.4cm) coordinate (t3greedy1);
\coordinate[circle, draw, inner sep = 1pt] (t3greedy1) at (t3greedy1) {};

\path (t2) ++(195:1.2cm) coordinate (t4);
\coordinate[rectangle, fill, draw, inner sep = 1.5pt] (t4) at (t4) {};

\coordinate[rectangle, fill, draw, inner sep = 1.5pt] (t5) at (-0.5, -0.5) {};

\coordinate[rectangle, fill, draw, inner sep = 1.5pt] (t6) at (0.1, -0.6) {};

\draw (t4)--(t2greedy2) -- (t2greedy1) -- (t2) -- (t2greedy3) -- (t5) -- (t6) -- (t1greedy2) -- (t3greedy1) -- (t3) (t1greedy2) -- (t1greedy1) -- (t1);

\draw (t2)--(t2fboundary) (t1)--(t1finterior) -- (t1fboundary) (t3)--(t3fboundary);

\path (6, 0.1) ++(80:1.1cm) coordinate (p0);

\path (6, 0.1) ++(90:1.1cm) coordinate (p1);
\path (6, 0.1) ++(110:1.1cm) coordinate (p2);
\path (6, 0.1) ++(130:1.1cm) coordinate (p3);
\path (6, 0.1) ++(150:1.1cm) coordinate (p4);
\path (6, 0.1) ++(170:1.1cm) coordinate (p5);
\path (6, 0.1) ++(190:1.1cm) coordinate (p6);
\path (6, 0.1) ++(210:1.1cm) coordinate (p7);
\path (6, 0.1) ++(230:1.1cm) coordinate (p8);
\path (6, 0.1) ++(250:1.1cm) coordinate (p9);
\path (6, 0.1) ++(270:1.1cm) coordinate (p10);
\draw [gray, fill=gray!12, color=gray!12] plot [smooth cycle] coordinates {(p0) (p1) (p2) (p3) (p4) (p5) (p6) (p7) (p8) (p9) (p10) (8, -1) (9.8, -0.4) (10.2, 0) (9.8, 0.5) (8.1, 1)};

\coordinate[rectangle, fill, draw, inner sep = 1.5pt] (s1) at (6, 0.1) {};
\draw[color = gray!100] (s1) circle (1cm){};
\path (s1) ++(-10:0.5cm) coordinate (s1greedy1);

\path (s1) ++(160:0.6cm) coordinate (s1finterior);
\coordinate[circle, draw, inner sep = 1pt, fill = gray!70] (s1finterior) at (s1finterior) {};
\path (s1) ++(190:1cm) coordinate (s1fboundary);
\coordinate[circle, draw, inner sep = 1.5pt, fill = gray!70] (s1fboundary) at (s1fboundary) {};
\coordinate[circle, draw, inner sep = 1pt] (s1greedy1) at (s1greedy1) {};
\path (s1) ++(10:1cm) coordinate (s1greedy2);
\coordinate[circle, draw, inner sep = 1pt] (s1greedy2) at (s1greedy2) {};

\coordinate[rectangle, fill, draw, inner sep = 1.5pt] (s2) at (8, -0.4) {};
\draw[color = gray!100] (s2) circle (0.4cm){};
\path (s2) ++(80:0.4cm) coordinate (s2greedy1);
\coordinate[circle, draw, inner sep = 1pt] (s2greedy1) at (s2greedy1) {};
\path (s2) ++(160:0.4cm) coordinate (s2fboundary);
\coordinate[circle, draw, inner sep = 1.5pt, fill = gray!70] (s2fboundary) at (s2fboundary) {};

\coordinate[circle, draw, inner sep = 1pt] (greedy1) at (7.5, 0.3) {};
\coordinate[rectangle, draw, fill, inner sep = 1.5pt] (s3) at (8.1, 0.7) {};
\coordinate[circle, draw, inner sep = 1pt] (greedy2) at (8.8, 0.3) {};
\coordinate[rectangle, draw, fill, inner sep = 1.5pt] (s4) at (9.8, 0.1) {};

\draw (s1)--(s1greedy1) -- (s1greedy2) -- (greedy1) -- (s2greedy1) -- (greedy2) -- (s4) (s2) -- (s2greedy1) -- (s3);
\draw (s1) -- (s1finterior) -- (s1fboundary) (s2) -- (s2fboundary);

\coordinate[rectangle, fill, draw, inner sep = 1.5pt, label = 70:$s$] (u1) at (1.85, 1) {};
\draw[dashed, color = gray!100] (u1) circle (1cm);

\coordinate[rectangle, fill, draw, inner sep = 1.5pt, label = -90:$t$] (u2) at (4.8, -0.6) {};
\draw[dashed, color = gray!100] (u2) circle (1cm);

\coordinate[circle, draw, inner sep = 1pt] (greedy5) at (3.2, 0.5) {};
\coordinate[circle, draw, inner sep = 1pt] (greedy6) at (4.1, -0.5) {};

\draw[dotted, thick] (t1fboundary) -- (u1) -- (greedy5) -- (greedy6) -- (u2) -- (s1fboundary){};

\end{tikzpicture}
    
    \caption{Augmented Greedy: Both dual balls around $s$ and $t$ of size $c$ intersect dual balls of size $c$. Connecting $s$ and $t$ by a greedy path, as well as connecting $s$ and $t$ to the respective dual balls they intersect, we reduce the number of connected components containing a dual ball of size $c$.}
    \label{fig:aug-greedy}
\end{figure}
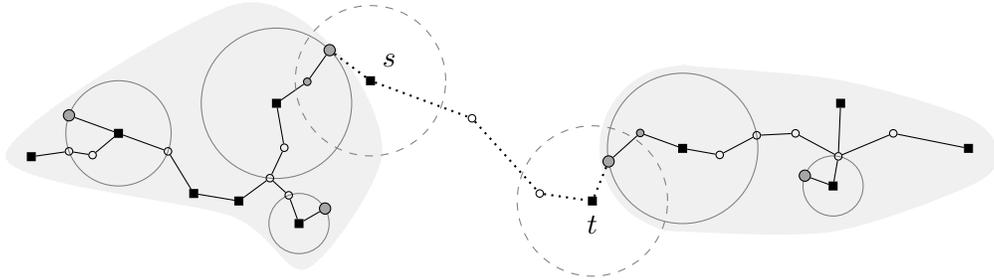

On arrival of a terminal pair $(s,t)$, we always purchase
the shortest path between $s$ and $t$, let $c$ denote its cost
and denote by $\lfloor\log c\rfloor$ the level of both $s$ and $t$.
If there are previous terminals $u$ and $v$ of the same level as $s$ and $t$ such that the distances
from $s$ to $u$ and from $t$ to $v$ are both at most $c/2$,
we make a step of augmented greedy \citep{BermanCoulston97}:
we also connect $s$ to $u$ and $t$ to $v$. This is of cost at most $O(c)$ and decreases the number of components
of that level, see Figure~\ref{fig:aug-greedy}.
If this is not the case, we would like to place
a dual ball of radius $\Omega(c)$ around either $s$ or $t$. 
This is possible in edge-weighted variant and, for each of the $\log k$ levels, such dual balls are disjoint and
serve as a lower bound on $\opt$.

However, the difficulty of the node-weighted setting comes from the fact that
a single vertex can be partially contained in many
dual balls. This means that such a vertex is useful for connecting
many terminals and we would like to choose few cheap such vertices
which are contained on boundaries of many dual balls. Similarly to \citet{HajiaghayiLiaghatPanigrahi13},
we use an instance of \nmfl treating such vertices as facilities to decide which ones we should purchase, see Figure~\ref{fig:boundary_nmfl} for an illustration.
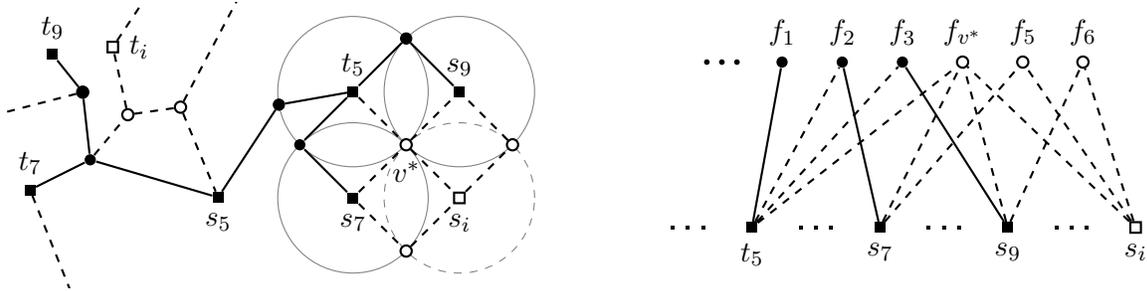
\begin{figure}[h]
\centering
\begin{subfigure}[t]{0.49\textwidth}
\centering
        \begin{tikzpicture}
    \clip (-5.4,-1.9) rectangle (2,1.9);
    \draw[draw=none] (-5.4,-1.9) rectangle (2,1.9);
    \coordinate[label = 90:$s_9$, rectangle, inner sep = 2pt, fill] (s9) at (45:1cm) {};   
    \coordinate[label = 90:$t_5$, rectangle, inner sep = 2pt, fill] (t5) at (135:1cm) {}; 
    \coordinate[label = -90:$s_7$, rectangle, inner sep = 2pt, fill] (s7) at (225:1cm) {}; 
    \coordinate[label = -90:$s_i$, rectangle, inner sep = 2pt, draw, thick] (si) at (315:1cm) {}; 
    \draw[gray!100] (s9) circle (1cm);
    \draw[gray!100] (t5) circle (1cm);
    \draw[gray!100] (s7) circle (1cm);
\draw[dashed, gray!100] (si) circle (1cm);
    \coordinate[label = -90:$s_5$, rectangle, inner sep = 2pt, fill] (s5) at (-2.5,-0.7) {}; 
    \coordinate[label = 90:$t_9$, rectangle, inner sep = 2pt, fill] (t9) at (-4.7,1.2) {}; 
    \coordinate[label = 90:$t_7$, rectangle, inner sep = 2pt, fill] (t7) at (-5,-0.6) {}; 
    \coordinate[label = 0:$t_i$, rectangle, inner sep = 2pt, draw, thick] (ti) at (-3.9,1.3) {}; 
    \coordinate[circle, inner sep = 1.5pt, fill] (v) at (-4.2,-0.2) {}; 
    \path (0,0) coordinate (vstar);
    \coordinate[label = -90:$v^*$, circle, inner sep = 1.5pt, draw, thick, fill=white] (vstar) at (vstar) {}; 
    \node[circle, inner sep = 1.5pt, thick, fill] (b_1) at (-180:1.414cm) {}; 
    \path (t5)++(190:1cm) coordinate (t5g);
    \coordinate[circle, draw, inner sep = 1.5pt, fill] (t5g) at (t5g) {};
    \coordinate[circle, inner sep = 1.5pt, draw, thick, fill=white] (greedy7i) at (-90:1.414cm) {}; 
    \coordinate[circle, inner sep = 1.5pt, draw, thick, fill=white] (greedy9i) at (0:1.414cm) {}; 
     \node[circle, inner sep = 1.5pt, draw, fill] (b_2) at (90:1.414cm) {}; 

     \coordinate[circle, draw, inner sep = 1.5pt, thick] (u1) at (-3, 0.5) {};
     \coordinate[circle, draw, inner sep = 1.5pt, thick] (u2) at (-3.7, 0.4) {};
     \coordinate[circle, draw, inner sep = 1.5pt, thick, fill] (u3) at (-4.3, 0.7) {};
     \draw[thick] (t7)--(v)--(u3)--(t9) (v)--(s5)--(t5g)--(t5)--(b_2)--(s9) (s7)--(b_1)--(t5);
     \draw[dashed, thick] (v)--(u2)--(u1)--(s5) (u1)--(-2.1, 2.2) (t7)--(-4.4, -2.1) (u2)--(ti) -- (-3.3, 2.2) (t5)--(vstar)--(si) (s7)--(vstar)--(s9) (s7)--(greedy7i)--(si)--(greedy9i)--(s9) (u3)--(-7,0);

    \end{tikzpicture}
    \caption{All nodes have weight $1$. Since $d_{G/S}(s_i, t_i) = 2$, we need to select a path and place a dual ball of radius~$1$ around~$s_i$. This intersects previously placed balls around $t_5, s_9, s_7$.}
    \label{fig:first}
\end{subfigure}
\hfill
\begin{subfigure}[t]{0.49\textwidth}
\centering
        \begin{tikzpicture}
    \draw[draw=none] (-5.4+0.125,-1.9) rectangle (2+0.125,1.9);
    \def\dist{0.85};
    \def\dots{0.2};
    \def\height{1.1}
\coordinate[rectangle, inner sep = 0.75pt, fill] (dot1) at (-5.4+\dist-\dots, -\height) {}; 
    \coordinate[rectangle, inner sep = 0.75pt, fill] (dot2) at (-5.4+\dist, -\height) {}; 
    \coordinate[rectangle, inner sep = 0.75pt, fill] (dot3) at (-5.4+\dist+\dots, -\height) {}; 
    \coordinate[rectangle, inner sep = 2pt, fill, label=-90:$t_5$] (t5) at (-5.4+\dist*2, -\height) {}; 
    \coordinate[rectangle, inner sep = 0.75pt, fill] (dot1) at (-5.4+3*\dist-\dots, -\height) {}; 
    \coordinate[rectangle, inner sep = 0.75pt, fill] (dot2) at (-5.4+3*\dist, -\height) {}; 
    \coordinate[rectangle, inner sep = 0.75pt, fill] (dot3) at (-5.4+3*\dist+\dots, -\height) {}; 
    \coordinate[rectangle, inner sep = 2pt, fill, label=-90:$s_7$] (s7) at (-5.4+\dist*4, -\height) {}; 
    \coordinate[rectangle, inner sep = 0.75pt, fill] (dot1) at (-5.4+5*\dist-\dots, -\height) {}; 
    \coordinate[rectangle, inner sep = 0.75pt, fill] (dot2) at (-5.4+5*\dist, -\height) {}; 
    \coordinate[rectangle, inner sep = 0.75pt, fill] (dot3) at (-5.4+5*\dist+\dots, -\height) {}; 
    \coordinate[rectangle, inner sep = 2pt, fill, label=-90:$s_9$] (s9) at (-5.4+\dist*6, -\height) {}; 
    \coordinate[rectangle, inner sep = 0.75pt, fill] (dot1) at (-5.4+7*\dist-\dots, -\height) {}; 
    \coordinate[rectangle, inner sep = 0.75pt, fill] (dot2) at (-5.4+7*\dist, -\height) {}; 
    \coordinate[rectangle, inner sep = 0.75pt, fill] (dot3) at (-5.4+7*\dist+\dots, -\height) {}; 
    \coordinate[rectangle, inner sep = 2pt, draw, thick, label=-90:$s_i$] (si) at (-5.4+\dist*8, -\height) {}; 

    \def\distf{0.8}
    \def\init{1.3}
    \coordinate[circle, inner sep = 0.75pt, fill] (dot1) at (-5.4+\init-0.2, \height) {}; 
    \coordinate[circle, inner sep = 0.75pt, fill] (dot1) at (-5.4+\init, \height) {}; 
    \coordinate[circle, inner sep = 0.75pt, fill] (dot1) at (-5.4+\init+0.2, \height) {}; 
    \coordinate[circle, inner sep = 1.5pt, fill, label=90:$f_1$] (f1) at (-5.4+\init+\distf*1, \height) {}; 
    \coordinate[circle, inner sep = 1.5pt, fill, label=90:$f_2$] (f2) at (-5.4+\init+\distf*2, \height) {};
    \coordinate[circle, inner sep = 1.5pt, fill, label=90:$f_3$] (f3) at (-5.4+\init+\distf*3, \height) {};
    \coordinate[circle, inner sep = 1.5pt, draw, thick, label=90:$f_{v^*}$] (fstar) at (-5.4+\init+\distf*4, \height) {};
    \coordinate[circle, inner sep = 1.5pt, draw, thick, label=90:$f_5$] (f5) at (-5.4+\init+\distf*5, \height) {};
    \coordinate[circle, inner sep = 1.5pt, draw, thick, label=90:$f_6$] (f6) at (-5.4+\init+\distf*6, \height) {};
\draw[thick] (t5)--(f1) (s7)--(f2) (s9)--(f3);
\draw[dashed, thick] (t5)--(f2) (t5)--(f3) (t5)--(fstar) (s7)--(fstar) (s7)--(f5) (s9)--(fstar) (s9)--(f6) (si)--(fstar) (si)--(f5) (si)--(f6);

    \end{tikzpicture}
    \caption{The \nmfl instance: Every arriving terminal corresponds to a client, every vertex on the boundary of its respective dual ball corresponds to a facility the client is connected~to.}
    \label{fig:second}
\end{subfigure}

\caption{The vertex $v^*$ is contained on the boundary of many dual balls. In other words, it would have been beneficial to select $v^*$. This corresponds to a cheap facility connected to many clients.}
\label{fig:boundary_nmfl}
\end{figure}

Our algorithm solves a single auxiliary instance of \nmfl, while
\citet{HajiaghayiLiaghatPanigrahi13} use a separate
instance of \nmfl for terminals at each level. Each one of these $\log k$ instances of \nmfl has an optimal solution of cost at most $\opt_\nwsf$, at most the cost of an optimal solution for the corresponding instance of \nwsf.
When provided with an $O(\log k \log n)$-competitive
algorithm to each of the \nmfl instances, their algorithm for \nwsf
is $O(\log^2 k \log n)$-competitive.
In contrast, we formulate a single \nmfl instance with a feasible solution of cost at most $O(\log k)\cdot\opt_\nwsf$. However, in our instance, we use facility costs scaled by a factor $\log k$. Hence, for each vertex corresponding to a facility selected, we only pay a $\smash{\tfrac{1}{\log k}}$ fraction of the cost. This way, 
having a $(O(\log k\log n), O(\log n))$-competitive solution
to the \nmfl instance, this results in a $O(\log k \log n)$-competitive solution for \nwsf.
We note that a similar idea of scaling the facility costs has been used already in \citep{HajiaghayiLiaghatPanigrahi14} in the context of \nwst.
Finally, we note that with a single instance of \nmfl, we slightly deviate from the charging scheme of~\cite{HajiaghayiLiaghatPanigrahi13}, as it is more involved to directly charge the greedy paths and facilities to dual balls or augmented greedy steps. We implicitly keep track of which level each facility opened corresponds to, and modify how augmented greedy steps are performed.

However, there is a fundamental issue when using a single instance of \nmfl; even against an oblivious adversary presenting the instance to $\nwsf$, we cannot handle the corresponding online instance of \nmfl by the classical randomized algorithm working against an oblivious adversary. Indeed, previous actions of our algorithm (e.g.\,which vertices are bought to connect previously arrived terminals) affect the distances of pairs arriving in the future, effectively changing the connections in the corresponding \nmfl instance. Of course, if we were to use a deterministic algorithm for \nmfl, this issue would not arise. In this case however, since then we require instance-awareness, the set of potential clients is polynomial in $n$ and results in a competitive ratio of $O(\log^2 n)$. Thus, to really obtain our refined bound of $O(\log k \log n)$, we need to use a randomized algorithm for \nmfl, and allow for adaptive changes in the client-facility graph. This setting is exactly captured by a semi-adaptive adversary: the $k$ terminal pairs fixed by an oblivious adversary for \nwsf implicitly determine a potential client-facility graph with $n$ facilities and $O(\log k \cdot k)$ clients. The decision which of these clients arrive is deferred to the adversary. Since the number of potential clients is polynomial in $k$, the $O(\log |C| \cdot \log|F|)$-competitive algorithm for \nmfl in this setting then guarantees the $O(\log k \log n)$-competitive ratio for \nwsf against an oblivious adversary.

\paragraph{Auxiliary instance of \nmfl}
We describe the auxiliary instance of \nmfl used by our algorithm.
Consider the graph $G=(V,E)$ with vertex weights $w \in \R_{\geq0}^V$.
An oblivious adversary has to fix an instance
of \nwsf beforehand. This instance is specified by the sequence
$\smash{T = ((s_1, t_1), \dotsc, (s_k, t_k)) \subseteq (V\times V)^k}$
of terminal pairs and determines a super-instance
$I(T)$ of \nmfl considered by our algorithm. The auxiliary instance $I_\nmfl$
will be chosen as a subinstance of $I(T)$ adaptively based on the previous steps of our online algorithm for \nwsf. Therefore the auxiliary instance $I_\nmfl$ can be solved by an algorithm for \nmfl against a semi-adaptive adversary. 
\begin{figure}[h]
\centering
\begin{subfigure}[t]{0.45\textwidth}
    \centering
        \begin{tikzpicture}
    \draw[white] (-3.5, -1.7) rectangle (3.5, 1.7);
    \clip (-3.5, -1.7) rectangle (3.5, 1.7);
        \coordinate[rectangle, inner sep=2pt, draw, thick, label=0:$t$] (s) at (0,0);
        \coordinate[circle, inner sep=1.5pt, draw, thick, label=-90:\small{$0.1$}] (v1) at (-1,-0.2);
        \coordinate[circle, inner sep=1.5pt, draw, thick, label=90:\small{$1.1$}] (v2) at (-2,0.7);
        \coordinate[circle, inner sep=1.5pt, draw, thick, label=180:\small{$100$}] (v3) at (-2.3,-0.5);
        \coordinate[circle, inner sep=1.5pt, draw, thick, label=90:\small{$4$}] (v4) at (0.8,+0.8);
        \coordinate[circle, inner sep=1.5pt, draw, thick, label=-90:\small{$1$}] (v5) at (1.7,+0.5);
        \coordinate[circle, inner sep=1.5pt, draw, thick, label=-0:\small{$1.5$}] (v6) at (0.6,-0.5);
        \coordinate[circle, inner sep=1.5pt, draw, thick, label=-180:\small{$3$}] (v7) at (1.2,-1.1);
        \draw[dashed, thick] (s)--(v1) --(v2)--(v3)--(-4,-3) (v1)--(v3) (s)--(v4)--(v5)--(4,4) (s)--(v6)--(v7)--(3,-2.5);
    \end{tikzpicture}
    \caption{Part of the node-weighted graph $G=(V,E)$.}
\end{subfigure}
\hfill
\begin{subfigure}[t]{0.45\textwidth}
    \centering
    \begin{tikzpicture}
    \def\height{0.9};
    
        \draw[white] (-3.5, -1.7) rectangle (3.5, 1.7);
        \clip (-3.5, -1.7) rectangle (3.5, 1.7);
        \coordinate[circle, inner sep=1.5pt, draw, thick, label=90:\small{$100\,\ell$}] (f100) at (-3, \height);
        \coordinate[circle, inner sep=1.5pt, draw, thick, label=90:\small{$1.1\,\ell$}] (f11) at (-2, \height);
        \coordinate[circle, inner sep=1.5pt, draw, thick, label=90:\small{$0.1\,\ell$}] (f01) at (-1, \height);
        \coordinate[circle, inner sep=1.5pt, draw, thick, label=90:\small{$1.5\,\ell$}] (f15) at (0, \height);
        \coordinate[circle, inner sep=1.5pt, draw, thick, label=90:\small{$3\,\ell$}] (f3) at (1, \height);
        \coordinate[circle, inner sep=1.5pt, draw, thick, label=90:\small{$4\,\ell$}] (f4) at (2, \height);
        \coordinate[circle, inner sep=1.5pt, draw, thick, label=90:\small{$1\,\ell$}] (f1) at (3, \height);
        \coordinate[rectangle, inner sep=2pt, draw, thick, label=-90:${(t,3)}$] (s0) at (-1.5, -\height);
        \coordinate[rectangle, inner sep=2pt, draw, thick, label=-90:${(t,4)}$] (s1) at (1.5, -\height);
        \draw[dashed, thick] (s0)--(f100) (s0)--(f11) (s0)--(f4) (s0)--(f15);
        \coordinate[label=-200:\small{$0.1$}] (c0100) at (-1.85, -0.7);
        \coordinate[label=90:\small{$0.1$}] (c0100) at (-1.45, -0.3);
        \coordinate[label=90:\small{$0$}] (c015) at (-0.85, -0.6);
        \coordinate[label=90:\small{$0$}] (c015) at (-0.7, -0.9);

        \draw[dashed, thick] (s1)--(f100) (s1)--(f4) (s1)--(f3);
        \coordinate[label=-200:\small{$0.1$}] (c1100) at (1.0, -1.05);
        \coordinate[label=-180:\small{$1.5$}] (c13) at (1.4, -0.35);
        \coordinate[label=0:\small{$0$}] (c14) at (1.60, -0.45);
    \end{tikzpicture}
    \caption{Part of the corresponding instance of NMFL.}
    \label{fig:corresponding_nmfl}
\end{subfigure}
\caption{Creating clients $(t,3)$, $(t,4)$ for the auxiliary instance of \nmfl, $\ell = \log k$. Facilities are on top with corresponding (scaled) opening costs, dashed lines are possible connections to facilities with their respective cost.}
\label{fig:creating_client_facility_graph}
\end{figure}
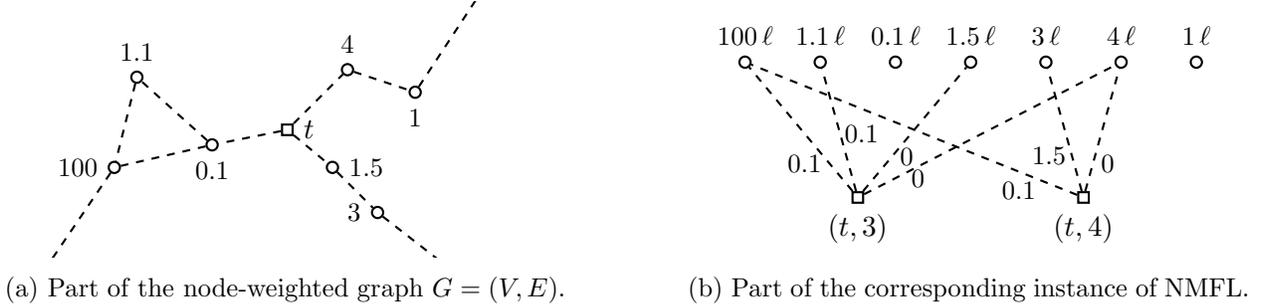

We will now precisely define $I(T) = (C, F, \cost)$.
For the set of facilities, we choose $F = V$, the whole vertex set of the graph $G$, each with facility opening costs $\fcost(v) = \ell \cdot w_v$, $\ell:=\log k$. For each of the $2\cdot k$ terminals $t$ contained in a pair in $T$, we create $\log k+1$ different clients $(t, 0), (t,1), \ldots, (t,\log k)$. This way, we have $|F| = n$ and $|C| = O(\log k\cdot k)$. It remains to define the connection costs between the clients and facilities. Client $(t,j)$ is connected to all facilities corresponding to vertices $v\in \Bd(t, 2^{j-3})$ and has connection cost $\ccost((t,j),v) = d_G(t,v)$. For all other vertices, i.e.\,for $v\notin \Bd(t,2^{j-3})$, we set $\ccost((t,j),v)=+\infty$. See Figure~\ref{fig:creating_client_facility_graph} for an illustration.

\subsection{Algorithm for \nwsf}\label{subsec:alg_nwsf}
We now give a formal overview of our algorithm. We assume that we are given $k$, the number of eventually arriving terminals pairs and that all arriving pairs have distance between $1$ and $k$. This is without loss of generality, see~\cref{subsec:guesskandopt}. For convenience, we also assume that $n$ and $k$ are powers of $2$ and we set $\ell := \log k$.

We maintain a solution $S \subseteq V$ such that $G[S]$ connects all arrived terminal pairs. Upon arrival of $(s,t)$ we compute the cheapest path $P$ between $s$ and $t$ in the graph $G/S$. Denoting by $c$ the weight of this path, we set $j := \lceil \log_2 c\rceil$
and add $s$ and $t$ to the set $T_j$ of terminals at level $j$. We purchase $P$ and proceed as follows.

First, we check whether the distance (w.r.t. $d_G(\cdot,\cdot)$) of $s$ to terminals at level $j$ is at least $2^{j -2}$ and no opened facility for a terminal at level $j$ is contained in the ball of radius $2^{j-2}$.
If this holds, we add the client $(s,j) \in C$ to $I_\nmfl$.
The online algorithm for \nmfl connects $s$ to a facility $f\in V$: we purchase the corresponding boundary vertex, add it to the set $F_j$, and purchase the corresponding connection. If this does not hold for terminal $s$ but for $t$, we proceed analogously with~$t$ replaced by~$s$.
Otherwise, we purchase the cheapest paths from $s$ and $t$ to $T_j \cup F_j$ -- both paths have cost at most $c$. At the end of the iteration, we add all purchased vertices to $S$.

\begin{algorithm}
\begin{algorithmic}[1]
\State Initialize $\mathcal{A}$ as Algorithm~\ref{alg:nmfl_to_sc} for \nmfl.\label{line:nwsf_nmflinit}
\When {$(s_i, t_i)$ arrives}
\State $P:=$ cheapest path from $s$ to $t$ in $G/S$
\State choose $j_i\in \N$ such that $\frac12\cdot 2^{j_i} < w(P) \leq 2^{j_i}$.
\State $T_{j_i} := T_{j_i} \cup \{s_i, t_i\}$.\\

\If{$d_{G}(s_i,T_{j_i})\geq 2^{j_i-2}$
    and $\Bd(s_i, 2^{j_i-3}) \cap F_{j_i} = \emptyset$}
\State Add client $(s_i, j_i)$ to $I_\nmfl$, let $\mathcal{A}$ connect it
  to some facility $f$, and add $f$ to $F_{j_i}$.
  \State Add $f$ and the cheapest path in $G/S$ between $s_i$ and $f$ to $S$.
   \label{line:if}
\ElsIf{$d_{G}(t_i,T_{j_i})\geq 2^{j_i-2}$
    and $\Bd(t_i, 2^{j_i-3}) \cap F_{j_i} = \emptyset$}
  \State Add client $(t_i, j_i)$ to $I_\nmfl$, let $\mathcal{A}$ connect it
  to some facility $f$, and add $f$ to $F_{j_i}$.
  \State Add $f$ and the cheapest path in $G/S$ between $t_i$ and $f$ to $S$.
    \label{line:elseif}
\Else
\Comment{Do the augmented greedy step}
\State Add the cheapest path in $G/S$ between $s_i$ and $T_{j_i} \cup F_{j_i}$ to $S$.\label{line:aug-greedy1}
\State Add the cheapest path in $G/S$ between $t_i$ and $T_{j_i} \cup F_{j_i}$ to $S$.\label{line:aug-greedy2}
\EndIf
\\
\State Add $P$ to $S$.\label{line:greedy_nwsf}
\EndWhen
\end{algorithmic}
\caption{Algorithm for \nwsf}
\label{alg:nwsf}
\end{algorithm}

\begin{lemma}
\label{lem:nmfl-opt}
Consider the instance $I_\nmfl$ constructed by Algorithm~\ref{alg:nwsf}.
The optimal solution to $I_\nmfl$ has cost at most $2\ell$ times
the optimal solution to the input instance of \nwsf.
\end{lemma}
\begin{proof}
Let $X \subseteq V$ be the optimal solution to the Steiner forest
problem of cost $\opt$. We construct a solution to the \nmfl
instance $I_\nmfl$ as follows:
\begin{itemize}
    \item Open each facility $v\in X$ for cost $\ell\cdot w_v$,
        paying $c(X) \leq \ell\cdot\opt$ in total.
    \item Connect each client $(s,j)$ in $I_\nmfl$ to facility $f(s,j)\in X$,
    as defined below,
    so that
    \[ \sum_{s\in T_j} \cost((s,j),f(s,j))\leq \opt
        \text{ for each $j=0,1,\dotsc, \ell$}.
    \]
\end{itemize}
This way, we obtain a feasible solution to $I_\nmfl$
of cost at most $2\ell\cdot\opt$.
It is enough to describe the mapping $f\colon C \to V$.

Fix $j \in \{1, \dotsc, \ell\}$. For each $s\in T_j$, there is a
$t$ such that $(s,t)$ is a pair of terminals
in the input instance $T$ of the Steiner forest and
$d_G(s,t) > 2^{j-1}$.
Therefore, there must be a path $P(s)\subseteq X \cap B(s, 2^{j-3})$
connecting $s$ to a vertex
$v \in X \cap \Bd(s,2^{j-3})$.
Let $f(s,j)$ be such a boundary vertex $v$.
Clearly, we have $\cost((s,j), f(s,j))\leq w(P(s))$.

Now we claim that $P(s) \cap P(s') = \emptyset$ for any $s,s' \in T_j$. 
To see why this is true, observe that whenever a client arrives that corresponds to some $s\in T_j$, it will have distance at least $2^{i_j-2}$ to all other $s' \in T_j$. 
So, if two clients arrive corresponding to $s, s'\in T_j$, then we have
$B(s, 2^{j-3}) \cap B(s', 2^{j-3}) = \emptyset$, which shows that $P(s) \cap P(s') = \emptyset$.

Hence, we have
\[ \sum_{s\in T_j} \cost((s,j), f(s,j))
    \leq \sum_{s\in T_j} w(P(s)) \leq w(F) = \opt.
\]
Summing over all $j\in \{0, 1, 2, \ldots, \ell\}$ then completes the proof.
\end{proof}

The competitive ratio of our algorithm depends on the performance of $\mathcal{A}$
on $I_\nmfl$.
This is formalized in the following lemma.

\begin{lemma}\label{lem:nwsf-competitive}
If algorithm $\mathcal{A}$ for \nmfl is $(\alpha, \beta)$-competitive
on $I_\nmfl$, then Algorithm~\ref{alg:nwsf} is
$O(1)(\alpha + \ell\cdot \beta)$-competitive.
\end{lemma}

In \cref{sec:nmfl} we present a randomized algorithm for online NMFL that is $(O(\log |C| \log |F|), O(\log |F|))$-competitive against semi-adaptive adversary that commits to a super-instance $I(T)=(F, C, \cost)$ beforehand
(see Theorem~\ref{thm:nmfl}).
This algorithm can be used as $\mathcal{A}$ in line~\ref{line:nwsf_nmflinit}
of Algorithm~\ref{alg:nwsf}, if Algorithm~\ref{alg:nwsf} is used
in a setting with oblivious adversaries: since an oblivious adversary
has to fix the input instance $T$ of \nwsf beforehand,
we can assume that the super-instance $I(T)$ for \nmfl is also
fixed beforehand.
By construction of $I(T)$, we have $|F| = n$ and $|C| = O(k\log k)$.
Lemma~\ref{lem:nwsf-competitive} then implies that
Algorithm~\ref{alg:nwsf} is $O(\log k \log n)$-competitive
against an oblivious adversary.
This proves Theorem~\ref{thm:intro_nwsf_rand} for the non-prize-collecting setting.

To derandomize Algorithm~\ref{alg:nwsf}, it is sufficient to use a deterministic algorithm for $\mathcal{A}$ -- indeed, in~\ref{subsec:online_nmfl} we also present a deterministic $(O(\log |C|\log |F|), O(\log |F|))$-competitive algorithm for \nmfl. For this deterministic algorithm however, it is crucial that a super-instance $I(T) = (C, F, \cost)$ be explicitly given in advance. The number of clients then depends on $|T|$, yielding a deterministic $O(\log|T|\log n)$-competitive algorithm. This proves \cref{thm:intro_nwsf_det} for the non-prize-collecting setting.

\begin{proof}[Proof of \cref{lem:nwsf-competitive}]
We split the cost incurred by Algorithm~\ref{alg:nwsf} in two
parts. The first part, denoted $C_1$ is incurred
in lines~\ref{line:if} and~\ref{line:elseif} due to following
the advice of the algorithm $\mathcal{A}$.
The second part, $C_2$, is incurred due to the greedy steps
in lines~\ref{line:aug-greedy1},~\ref{line:aug-greedy2} and~\ref{line:greedy_nwsf}.

To bound the expected value of $C_1$, it is enough to observe that the cost
paid by $\mathcal{A}$ for each facility $v$ is $\ell\cdot w_v$,
while Algorithm~\ref{alg:nwsf} pays only $w_v$.
Denoting $\conn(\mathcal{A})$ and $\fac(\mathcal{A})$ the connection and
facility costs respectively incurred by $\mathcal{A}$, we have
\[
\textstyle
\mathbb{E}[C_1]=\mathbb{E}[\conn(\mathcal{A}) + \frac1\ell \fac(\mathcal{A})]
\leq (\beta + \frac{\alpha}{\ell})\cdot \opt(I_\nmfl)
\leq 2(\ell\cdot \beta + \alpha)\cdot\opt,
\]
The first inequality follows from the assumption
about the performance
of $\mathcal{A}$, the last one from Lemma~\ref{lem:nmfl-opt}.

Now we prove a similar bound for $C_2$.
In iteration $i=1, \dotsc, k$, the cost incurred
in each line \ref{line:aug-greedy1}, \ref{line:aug-greedy2} and \ref{line:greedy_nwsf}
is bounded by $8\cdot 2^{j_i}$.
Denoting $I_1$ the set of iterations where one of the lines
\ref{line:if} or \ref{line:elseif} were executed,
and $I_2 = \{1, \dotsc, k\} \setminus I_1$ the iterations where
lines \ref{line:aug-greedy1} and \ref{line:aug-greedy2} were executed,
we can write
\begin{equation}
\label{eq:lem_alg_C2}
C_2 \leq \sum_{i\in I_1} 2^{j_i} + 3\cdot \sum_{i\in I_2} 2^{j_i}.
\end{equation}
At each iteration $i \in I_1$ we add one facility to $F_{j_i}$
and therefore we have
\begin{equation}
\label{eq:lem_alg_I}
\sum_{i\in I_1} 2^{j_i} \leq \sum_{j=0}^{\ell} |F_j|\cdot 2^{j}.
\end{equation}
Similarly, we claim that
\begin{equation}
\label{eq:lem_alg_II}
\sum_{i\in I_2} 2^{j_i} \leq \sum_{j=0}^{\ell} |F_j| \cdot 2^{j}
\end{equation}
which requires an argument.
Consider an iteration $i\in I_2$,
and paths $P(s_i)$ and $P(t_i)$ purchased
in lines \ref{line:aug-greedy1}, \ref{line:aug-greedy2}
respectively. First, note that
$w(P(s_i)) < 2^{j_i-2}$, since
$B(s_i, 2^{j_i-2}) \cup \Bd(s_i, 2^{j_i-3})$
contains either a terminal from $T_{j_i}$ or a facility
from $F_{j_i}$. Similarly, we have $w(P(t_i)) < 2^{j_i-2}$.
If the endpoints of $P(s_i)$ and $P(t_i)$ were connected
in $G[S]$, we would have
$d_{G/S}(s_i,t_i) < 2^{j_i-2} + 2^{j_i-2} \leq d_{G/S}(s_i,t_i)$
-- a contradiction.
Therefore, we denote
$C(s_i)$ the component of $G/S$ containing the endpoint
of $P(s_i)$ and $C(t_i)$ the component containing the endpoint of $P(t_i)$.
Each of them contains at least one facility from $F_{j_i}$.
In particular,
consider the smallest $i'$ such that $s_{i'}$ and $t_{i'}$
belong to $C(s_i)$ and $j_{i'} = j_i$. Then $i'$ belongs to $I_1$ and
either $s_{i'}$ or $t_{i'}$
must be connected to some facility $v\in F_{j_i}$.
The same argument applies to $C(t_i)$.
Therefore, each iteration $i \in I_2$ decreases the number of connected components of $G[S]$ containing a facility from $F_{j_i}$
and \eqref{eq:lem_alg_II} follows.

Let us denote
$F' = \bigcup_{j=0}^{\ell} F_j$ and $j(v) = \max\{j\mid v\in F_j\}$
for each $v\in F$.
Together with
\eqref{eq:lem_alg_C2}, \eqref{eq:lem_alg_I}, and \eqref{eq:lem_alg_II},
we have
\[ C_2 \leq O(1) \sum_{j=0}^\ell |F_j| 2^j
    \leq O(1) \sum_{v\in F'} \sum_{j=0}^{j(v)} 2^{j}
    \leq O(1) \sum_{v\in F'} 2^{j(v)}.
\]
Each $v\in F'$ was opened by algorithm $\mathcal{A}$ for price
$\ell\cdot w_v$ and connected (possibly in a different iteration)
to some client $(s, j(v))$, such that
$v\in \Bd(s, 2^{j(v)-2})$. Therefore we have
$2^{j(v)-2} \leq d_G(s,v) + w_v = \cost((s,j(v)), v) + \cost(v)/\ell$.
This implies
$\sum_{v\in F'} 2^{j(v)-2} \leq O(\conn(\mathcal{A}) + \fac(\mathcal{A})/\ell)$
and the lemma follows.
\end{proof}

\subsection{Extending to the prize-collecting setting}\label{subsec:prize-collecting}\label{subsec:alg_pcnwsf}
To extend an algorithm for a node-weighted problem to the prize-collecting setting, \cite{HajiaghayiLiaghatPanigrahi14} use an additional online algorithm to the fractional version of the prize-collecting problem. Based on the output of this algorithm, it is decided which pairs should be connected. The connections are then made by the original algorithm. This approach adds an extra $\log(n)$ factor to the competitive ratio.

We propose a different approach, that exploits the non-metric facility location problem that is already used inside our algorithm. This way, no extra factor in the competitive ratio is incurred. We modify the auxiliary \nmfl instance in the following way:
one extra facility $f_0$ of cost zero is added to the set of available facilities, and for each arriving pair $(s_i,t_i)$ and all levels $j\in [\log k]$, the clients $(s_i,j)$ and $(t_i,j)$ can be connected to $f_0$ at cost $p_i$. The algorithm now proceeds exactly as in the non-prize-collecting setting except for one difference: Whenever the algorithm for \nmfl decides to connect $(s_i, j)$ or $(t_i, j)$ to $f_0$, the algorithm halts and pays the penalty $p_i$ instead (in this case, we also do not add $s_i$ and $t_i$ to the corresponding set $T_{j_i}$ nor add $f_0$ to $F_j$). We call $\inst'_G$ the resulting instance of \nmfl. The rest is identical to \cref{alg:nwsf}, see \cref{alg:pcnwsf} for a description.

\begin{algorithm}
\begin{algorithmic}[1]
\State Initialize $\mathcal{A}$ as Algorithm~\ref{alg:nmfl_to_sc} for \nmfl.\When {$(s_i, t_i)$ arrives}
\State $P:=$ cheapest path from $s$ to $t$ in $G/S$
\State choose $j_i\in \N$ such that $\frac12\cdot 2^{j_i} < w(P) \leq 2^{j_i}$.

\If{$d_{G}(s_i,T_{j_i})\geq 2^{j_i-2}$
    and $\Bd(s_i, 2^{j_i-3}) \cap F_{j_i} = \emptyset$}
\State Add client $(s_i, j_i)$ to $I_\nmfl$, let $\mathcal{A}$ connect it
  to some facility $f$.
  \State If $f=f_0$, pay penalty $p_i$ and continue to the next iteration.
\ElsIf{$d_{G}(t_i,T_{j_i})\geq 2^{j_i-2}$
    and $\Bd(t_i, 2^{j_i-3}) \cap F_{j_i} = \emptyset$}
  \State Add client $(t_i, j_i)$ to $I_\nmfl$, let $\mathcal{A}$ connect it
  to some facility $f$.
  \State If $f=f_0$, pay penalty $p_i$ and continue to the next iteration.
\Else
\Comment{Do the augmented greedy step}
\State Add the cheapest path in $G/S$ between $s_i$ and $T_{j_i} \cup F_{j_i}$ to $S$.\label{line:aug-greedy1PC}
\State Add the cheapest path in $G/S$ between $t_i$ and $T_{j_i} \cup F_{j_i}$ to $S$.\label{line:aug-greedy2PC}
\EndIf
\\
\State Add $P$ to $S$.\EndWhen
\end{algorithmic}
\caption{Algorithm for \pcnwsf}
\label{alg:pcnwsf}
\end{algorithm}

To prove a competitive ratio of $O(\log k\log n)$ for \cref{alg:pcnwsf} we claim that \cref{lem:nmfl-opt,lem:nwsf-competitive} still hold.
\begin{lemma}
  The optimal solution to $\mathcal{I}'_G$ is at most $2\ell$ times the optimal solution to the \pcnwsf problem. \label{lem:pcnmfl-opt}
\end{lemma}
\begin{proof}
Consider an optimal solution $S$ to the \pcnwsf problem. Let $H$ be the set of indices $i$ of pairs $(s_i,t_i)$ that are connected in $S$. The objective value $\opt$ of $S$ consists of the cost $C_\text{nodes}$ of bought nodes plus the penalty cost of the unconnected pairs, i.e. $\opt = C_\text{nodes} + \sum_{i\notin H} p_i$. In the proof of \cref{lem:nmfl-opt} we have shown how to construct a NMFL solution that serves all arriving clients $(s_i,j_i)$ or $(t_i,j_i)$ for $i\in H$ at cost at most $2\ell C_\text{nodes}$.

The remaining clients are of the form $(s_i,j_i)$ or $(t_i,j_i)$ for $i\notin H$. We connect these clients to facility $f_0$ at cost $p_i$. The sum of these costs is at most $\sum_{i\notin H} p_i$. So, the total cost of this solution to NMFL is at most $2\ell C_\text{nodes}+\sum_{i\notin H} p_i \leq (2\ell) \cdot \opt$.
\end{proof}

\begin{lemma}\label{lem:pc-nwsf-competitive}
  If algorithm $\mathcal{A}$ is $(\alpha, \beta)$-competitive on the online facility location instance constructed in \cref{alg:pcnwsf}, then \cref{alg:pcnwsf} is $O(1)(\alpha + \ell\cdot \beta)$-competitive.
\end{lemma}
\begin{proof}
We first consider the cost incurred through penalties, $\sum_{i\notin H}p_i$. Since penalties correspond to connection costs in $\inst_G'$ and algorithm $\mathcal{A}$ is $\beta$-competitive with respect to these,
the cost of penalties paid by Algorithm~\ref{alg:pcnwsf} is never more than $O(\beta)\cdot\opt_\nmfl \leq O(\ell\cdot\beta)\cdot \opt_\pcnwsf$
by Lemma~\ref{lem:pcnmfl-opt}.

We now bound the cost incurred over all the remaining iterations. Note that for these pairs of terminals we proceed analogously to Algorithm~\ref{alg:nwsf}, i.e.\,we always connect one of the terminals in the pair to a facility on its boundary or do augmented greedy. We bound these cost according to Lemma~\ref{lem:nwsf-competitive}. In particular, we bound the cost incurred by the greedy and augmented greedy connection steps by $O(\alpha + \ell\cdot \beta)$ times the optimum value for \nwsf restricted to the client pairs $(s_i, t_i)$, where $i\in H$ (the clients where Algorithm~\ref{alg:pcnwsf} opts not to pay penalties). The latter is never more than $\opt_\pcnwsf$, hence the lemma follows.
\end{proof}

As in the previous section, we can now plug in our $(O(\log |C|\log |F|), O(\log |F|))$-competitive algorithm for \pcnwsf in the semi-adaptive setting. Just like before, this gives an $O(\log k \log n)$ competitive algorithm for prize-collecting NWSF. Similarly, using a deterministic algorithm for \nmfl, we obtain~\cref{thm:intro_nwsf_det}, a deterministic $O(\log |T|\log n)$-competitive algorithm for \pcnwsf.

\section{Semi-adaptive online non-metric facility location}\label{sec:nmfl}

In this section we will introduce an algorithm for online \nmfl that works against semi-adaptive adversaries.
A semi-adaptive adversary needs to commit to an superinstance $(F,C,\cost)$ of NMFL at the beginning of its run. In particular it needs to decide the prices of the facilities, which clients each facility is connected to and the corresponding connections costs. This information is not revealed to the algorithm yet. In each timestep the adversary can then adaptively decide which of the clients will arrive. We prove the following theorem, which is a fine-grained version of \cref{thm:intro_nmfl}.
\begin{theorem}
\label{thm:nmfl}
There is an $(O(\log |F| \log |C| ), O(\log |C| ))$-competitive randomized algorithm for online non-metric facility location
against semi-adaptive adversaries. If the client-facility graph is known in advance, the online algorithm can be made deterministic. 
\end{theorem}

The core of our proof is in showing that the well-known algorithm for rounding fractional solutions for online set cover still works in the semi-adaptive setting. This we do in \cref{subsec:rounding_sc}. Then we will show in \cref{subsec:online_nmfl} that one can reduce the online NMFL problem to rounding fractional solutions for online set cover.

\subsection{Online set cover rounding against semi-adaptive adversaries}
\label{subsec:rounding_sc}

We propose a rounding procedure for set cover in setting with semi-adaptive
adversaries, losing a factor $O(\log |\X|)$ on top of the cost of the
provided fractional solution.
This procedure is almost identical to
\citep{alon_general_2006} with the number of rounding thresholds being the
only difference.
For each set $S$, we sample a random variable $Y_S$ as the minimum of $p := \lceil 2\cdot \log|\X|\rceil$ independent random variables that are uniformly distributed in $[0,1]$.
We then round the fractional solution $\{x_S\}_{S\in\S}$ by including set $S$ as soon as $x_S \geq Y_S$.
If this procedure fails, i.e. fails to cover the element $e$ that has just arrived, we greedily buy the cheapest set covering $e$.
This procedure is summarized in Algorithm~\ref{alg:rounding_SC_semiadaptive}.
Note that Algorithm~\ref{alg:rounding_SC_semiadaptive} requires knowledge of
$\log|\X|$ up to a multiplicative factor of~$2$.
This is without loss of generality as shown in \cref{app:doubling_size_guess}.

\begin{algorithm}
\begin{algorithmic}[1]
\State parameter $p:=$ our guess of $4\log |\X|$ \Comment{Can be found by a doubling procedure, see Appendix~\ref{app:doubling_size_guess}.}
\State For each $S\in\S$, sample $Y_S$ as the minimum of $p$ independent random variables uniform in $[0,1]$.\label{alg:sample_Y_S}
\When {$(e_i, \S(e_i), \{x^{(i)}_S\}_{S\in\S})$ arrives}
\State Purchase all sets $S$ with $x_S \geq Y_S$.\label{line:buy_proportional_YS}
\If{$e_i$ remains uncovered}\label{line:remains_uncovered}
\State{Buy cheapest set $S \in \S(e_i)$.}\label{line:buy_cheapest}
\EndIf
\EndWhen
\end{algorithmic}
\caption{Rounding set cover against a semi-adaptive adversary\label{alg:rounding_SC_semiadaptive}}
\end{algorithm}

The following two claims formalize crucial parts
of the analysis of Algorithm~\ref{alg:rounding_SC_semiadaptive}.
While similar statements are relatively easy to prove
in the setting with oblivious adversaries,
the adaptiveness of the adversary brings considerable
difficulties to the analysis.

\begin{claim}\label{claim:proportional_x_sc}
Let $Y_S$ be a random variable sampled in Line~\ref{alg:sample_Y_S}.
Given the final fractional solution $x_S^{(|X|)}$, we have
    \begin{displaymath}
        \Pr[x_S^{(|\X|)}\geq Y_S] \leq (2\cdot p)\cdot x_S^{(|\X|)}.
    \end{displaymath}
\end{claim}

The proof of the above claim is trivial in the setting with oblivious adversaries\footnote{Since $x^{(|X|)}$ is deterministically derived from the input instance which
has to be fixed by an oblivious adversary beforehand,
the probability that $U\sim \mathrm{Unif}([0,1])$ is at most $x_S$ is equal to $x_S$
and $\Pr[Y_S \leq x_S^{(|X|)}] \leq p \cdot x_S^{(|X|)}$.
}.
However, a semi-adaptive adversary can see from the algorithm's actions
whether $x_S$ already surpassed $Y_S$ and strategically increase
fractional variables which are most likely to surpass their respective thresholds next.

\begin{claim}\label{claim:bad_event_SC}
For each $i=1, \dotsc, |X|$,
consider event $E_i$ that $e_i$ triggers Line~\ref{line:buy_cheapest}. 
We have
\begin{alignat}{1}
    \Pr[E_i]\leq \exp(-p). \label{claim:event_bound}
\end{alignat}
\end{claim}

In the setting with oblivious adversaries, each requested element
$e_i$ has to be fixed by adversary beforehand
and, in particular, events $E_i$ are then independent which
makes the analysis much simpler.
A semi-adaptive adversary can see that we have failed
to select any set covering certain element and choose new
elements to arrive based on this fact.
We postpone the proofs of both claims to the later parts of this section.

The following lemma summarizes the performance of our algorithm in a setting
with an additional restriction on the behavior of the semi-adaptive adversary.
Apart from the fact that the adversary has to fix the super-instance
$(X, \S)$ beforehand
and is allowed to request only elements from $X$, we also impose
an upper bound on the cost of the feasible fractional solution
provided by the adversary.
This condition is not very restrictive and Lemma~\ref{lem:fixed-semi-adaptive-sc}
turns out to be the main part of the proof of
our result for set cover
(Theorem~\ref{thm:semi-adaptive-sc} below).
In particular, such $B$ is naturally determined
for the \nmfl instance produced by Algorithm~\ref{alg:nwsf}
depending on the optimum cost of the \nwsf instance.

\begin{lemma}
Let $(X, \S)$ be a super-instance and $B\geq 0$ a parameter fixed beforehand.
Consider an adversary which adaptively
selects $X' \subseteq X$ and provides the online algorithm with a monotone, fractional solution
of cost at most $B$.
Then Algorithm~\ref{alg:rounding_SC_semiadaptive} satisfies
\[ \Exp[\alg(X',\S)] \leq O(\log |X|)\cdot B.\]
    \label{lem:fixed-semi-adaptive-sc}
\end{lemma}

The following theorem is the main result of this section.
While Lemma~\ref{lem:fixed-semi-adaptive-sc} bounds the expected cost
of $\alg$ using $B$,
Theorem~\ref{thm:semi-adaptive-sc} requires a bound in terms of
$\Exp[c^Tx]$. To control the evolution of $c^Tx$,
we consider a geometric series of thresholds.
Whenever the cost of the fractional solution exceeds a threshold $b$,
we double the threshold $b$ and restart Algorithm~\ref{alg:rounding_SC_semiadaptive}.
This way, we have a lower bound on $c^Tx$ which
can be combined with Lemma~\ref{lem:fixed-semi-adaptive-sc}.

\begin{theorem}\label{thm:rounding_fracSC}
Consider a semi-adaptive adversary which fixes a super-instance $(X,\S)$ and a cost vector $c:\S\rightarrow \R_{\geq 0}$
beforehand and provides a monotonically-increasing
feasible fractional solution $x \in [0,1]^{\S}$.
There is an online rounding algorithm against such adversary
which produces an integral solution to $(X', \S)$ of expected cost at most $O(\log|X|) \cdot \Exp[c^T x]$.
Consequently, there is a $O(\log|X|\cdot\log|\S|)$-competitive online algorithm for set cover against a semi-adaptive adversary.
    \label{thm:semi-adaptive-sc}
\end{theorem}

\begin{proof}[Proof of \cref{thm:semi-adaptive-sc}]
Consider the following algorithm: When the first element and its corresponding fractional solution arrives,
choose $b_1 = 2\cdot c^T x^{(1)}$.
Having a threshold $b_j$, run \cref{alg:rounding_SC_semiadaptive} until step $t$
when $c^T x^{(t)} > b_j$. Then set $b_{j+1} = 2\cdot c^T x^{(t)}$
and continue with the iteration $j+1$.
Note that $b_{j+1} \geq 2b_j$ for each $j =1, \dotsc, r$ where $r$ is the number
of iterations of this loop.
Let $C_j$ be the cost of $\alg$ during iteration $j$.
We have
\begin{equation}\notag
\Exp[\alg] \leq \Exp\bigg[ \sum_{j=1}^r C_j\bigg]
    = \sum_{j=1}^{\infty} \Pr[r\geq j] \cdot \Exp[C_j \mid r\geq j]
    \leq O(\log |X|)\cdot \sum_{j=1}^{\infty} \Pr[r\geq j] \cdot b_j,
\end{equation}
because, if iteration $j$ is executed, $\Exp[C_j] \leq O(\log |X|) \cdot b_j$
by Lemma~\ref{lem:fixed-semi-adaptive-sc}.
We can rearrange the sum in the right-hand side to get the following bound:
\begin{equation}\notag
\sum_{j=1}^{\infty} \Pr[r\geq j] \cdot b_j
    = \sum_{j=1}^{\infty}\bigg( \Pr[r = j] \sum_{i=1}^j b_i\bigg)
    \leq \sum_{j=1}^{\infty} \Pr[r = j]\cdot 2b_j
    = 2\Exp[b_r]
    \leq 4\Exp[c^Tx],
\end{equation}
where $x$ is the final fractional solution at the end of the input.
The first inequality holds because the $b_j$ are geometrically increasing,
and the last inequality holds because $c^Tx \leq 2b_r$
by definition of $r$.

The preceding two inequalities imply that $\Exp[\alg] \leq O(\log|X|) \Exp[c^Tx]$.
There is a deterministic fractional algorithm for set cover by \citet{alon_general_2006} which does not
require knowledge of $(X, \S)$ beforehand and
is $O(\log |\S|)$-competitive. Their result
implies the existence of a
$O(\log |X| \log |\S|)$-competitive algorithm for set cover
against a semi-adaptive adversary.
\end{proof}
\begin{proof}[Proof of \cref{lem:fixed-semi-adaptive-sc}]
Let $(\X, \S)$
denote the super-instance fixed by the adversary beforehand
and $e_1, e_2, \dotsc, e_{|X|}$ be the elements of $X$.
We consider $|X|$ times steps and assume that,
for each time step $i=1, \dotsc, |X|$, the adversary must irrevocably decide whether
or not to include $e_i \in X$ in the real instance.
This assumption is without loss of generality, up to a quadratic increase
in the size of $X$:
Consider a set system $(\X', \S)$, where $\X'$ has size $|X|^2$ and
contains $|X|$ copies $\smash{e_i^1, \dotsc, e_i^{|X|}}$ of each element $e_i\in X$,
each belonging to the same sets as in $(\X, \S)$.
Any ordering of the elements from $\X$ chosen by the adversary then corresponds to a selection of the elements of $\X'$, where the elements have to be processed in order. 
The online algorithm sees the same input regardless whether the adversary's
super-instance is $(\X',\S)$ or $(\X, \S)$.

We bound the costs incurred in lines \ref{line:buy_proportional_YS}
and \ref{line:buy_cheapest} separately.
To bound the cost incurred in line~\ref{line:buy_proportional_YS}, we use Claim~\ref{claim:proportional_x_sc}.
By linearity of expectation, Claim~\ref{claim:proportional_x_sc} implies
\begin{alignat}{1}
    \Exp[\text{cost of sets bought in Line~\ref{line:buy_proportional_YS}}] &\leq  O(p)\cdot \Exp[\sum c_S\cdot x_S^{(|\X|)}]\notag\\
    &\leq O(\log|\X|\cdot\log|\S|)\cdot \Exp[\opt].\label{bound:set_cover_sample_relative_x_S}
\end{alignat}

Now we show that the probability of purchasing a set in Line~\ref{line:buy_cheapest}
is small.
Consider time step $t$ when the adversary decides whether element $e_t$ arrives. Let $x^{(t)}$ and  $\hat{x}^{(t)}$ be the fractional solution and the integral solution at the end of time step $t$. Let $A_i$ denote the event that element $e_i$ arrives and let $C_i^{(t)}$ be the event that $e_i$
is covered by $\hat{x}^{(t)}$. Note that $E_i=A_i\setminus C_i^{(i)}$.
Claim~\ref{claim:bad_event_SC} says that the probability
of $E_i$ is at most $\exp(-p)$.

Let $\hat{E}$ be the event that any of the events $E_i$ is triggered. By taking the union bound over all possible $i = 1, \dotsc, |\X|$, we see that
$\Pr[\hat{E}]\leq \exp(-p) \cdot |X|$.
Each time that an event $E_i$ happens, we pay the cost of the cheapest
set containing $e_i$ which is at most $c^Tx^{(i)}\leq B$.
So, when $\hat{E}$ happens we incur a cost of at most $|\X|\cdot B$.
Hence, we need to choose $p$ of order $\log|X|$ in order to get
\begin{alignat}{1}
  \Exp[\text{cost of sets bought in Line~\ref{line:buy_cheapest}}] &\leq \Pr[\hat{E}]\cdot |\X|\cdot B \leq B .\label{bound:set_cover_bad_event}
\end{alignat}

Combining~(\ref{bound:set_cover_sample_relative_x_S}) and~(\ref{bound:set_cover_bad_event}), the expected cost of the algorithm is at most $O(\log|\X|\log|\S|)\cdot \Exp[\opt]$.\end{proof}

\begin{proof}[Proof of Claim~\ref{claim:proportional_x_sc}]
We show that $\Pr[x_S^{(t)} \leq U] \leq 2x_S^{(t)}$ whenever $U \sim \mathrm{Unif}([0,1])$.
The claim then follows from the union bound, since
$Y_S := \min(U_1, \ldots, U_p)$, where $U_i\sim\mathrm{Unif}([0,1])$
for each $i=1, \dotsc, p$.

Consider a set $S$ which was not purchased by the algorithm
until time $t-1$ and the adversary increases its fractional value by
$\smash{\delta_S^{(t)} = x_S^{(t)} - x_S^{(t-1)}}$.
The probability that $x_S^{(t)}$ rises above the threshold~$U$ is then
$\smash{\nicefrac{\delta_S^{(t)}}{1-x_S^{(t-1)}} \leq 2\delta_S^{(t)}}$
for $\smash{x_S^{(t)} < 1/2}$ (Note that the claim is trivially true
for $\smash{x_S^{(t)} \geq 1/2}$).
In what follows, we make this intuition formal
by analyzing the adaptive decisions of the adversary.

Taking the contrapositive, we are going to lower bound $\smash{\Pr[x_S^{(t)} < U] = \Exp[\1_{x_S^{(t)} < U}]}$.
Set $\delta_S^{(t)} := x_S^{(t)}-x^{(t-1)}_S$ and denote by $\mathscr{F}_{t-1}$ the filtration generated by $x^{(1)}, x^{(2)}, \ldots, x^{(t-1)}$.
First, we compute $\Exp[\1_{x_S^{(t)} < U}]$ conditioned
on the filtration $\mathscr{F}_{t-1}$.
We have
    \begin{align*}
        \Exp[\1_{x_S^{(t)} < U}|\mathscr{F}_{t-1}] &= \1_{x_S^{(t-1)}<U}\cdot\Pr[x_S^{(t)}<U\mid x_S^{(t-1)} < U]\\
        &=\1_{x_S^{(t-1)} < U}\cdot\tfrac{1-x_S^{(t-1)}-\delta_S^{(t)}}{1-x_S^{(t-1)}}\\
        &=\1_{x_S^{(t-1)} < U}\cdot(1-\tfrac{\delta_S^{(t)}}{1-x_S^{(t-1)}})
    \end{align*}
The computation above implies that the following random variable
is a martingale:
\begin{equation}
\label{eq:sc_martingale1}
    Z_i^{(t)}:= \1_{x_S^{(t)} < U}\cdot \prod_{j=1}^{t}(1-\tfrac{\delta_S^{(j)}}{
1-x_S^{(j-1)}
    })^{-1}
\end{equation}
Hence, we have $\Exp[Z_i^{(|X|)}]=\Exp[Z_i^{(0)}]=1$ and
given a final fractional solution $x^{(|X|)}$, 
we can bound the probability of $x_S^{(|X|)}$ being smaller than $U$ as
\begin{displaymath}
    \Exp[\1_{x_S^{(|\X|)} < U}] = \prod_{j=1}^{|\X|}(1-\tfrac{\delta_S^{(j)}}{
1-x_S^{(j-1)}
    }).
\end{displaymath}
Now it is enough to use the inequality
$
    \prod (1-x_i) \geq 1-\sum x_i,
$
which holds for all $x_i \in [0,1)$. We get
\begin{displaymath}
    \Exp[\1_{x_S^{(|\X|)} < U}] \geq \begin{cases}
        1-\sum_{i=1}^{|\X|}2\cdot \delta_S^{(i)} &\text{ if }
x_S^{(|\X|-1)} < 1/2,\\
        0 &\text{ otherwise.}
    \end{cases}
\end{displaymath}
This implies that $\Pr[x_S^{(t)} \leq U] \leq 2x_S^{(t)}$.
\end{proof}

In order to prove Claim~\ref{claim:bad_event_SC},
we analyze the behaviour of the adversary using a martingale similar to~\eqref{eq:sc_martingale1}.
We use the following classical fact from the theory of martingales
called Doob's optional stopping theorem.
While the expected value of a martingale at any fixed time $t$
is its initial value, Doob's theorem says that the same holds
for $t$ chosen by an arbitrary stopping strategy.

\begin{proposition}[see page 555, Theorem 9 in \citep{GrimmettStirzaker20}]
\label{prop:stopping}
Let $(Z, \mathscr{F})$ be a martingale and let $T$ be a stopping time.
Then $\Exp[Z_T] = \Exp[Z_0]$ if the following holds:
\begin{itemize}
    \item $\Pr[T < \infty] = 1$, $\Exp[T] < \infty$, and
    \item there exists a constant $c$ such that $\Exp\big[|Z_{n+1} - Z_n|\, \big\vert\, \mathscr{F}_n\big] \leq c$ for all $n < T$.
\end{itemize}

\end{proposition}

\begin{proof}[Proof of Claim~\ref{claim:bad_event_SC}]
Let $\mathscr{F}_t$ be the filtration generated by
arrivals of elements
$A_{1},\ldots, A_{t+1}$
and fractional solutions
$x^{(1)},\ldots, x^{(t+1)}$ until time $t+1$ and
the integral solution produced by Algorithm~\ref{alg:rounding_SC_semiadaptive}
$\hat{x}^{(1)},\ldots,\hat{x}^{(t)}$ until time $t$.
Conditioned on $\mathscr{F}_{t-1}$, we estimate the probability of
the event $\neg C_i^{(t)}$ that an element $e_i$ is not covered
by $\hat{x}^{(t)}$. We have
  \begin{align*}
    \Exp[\1_{\neg C_i^{(t)}}\mid \mathscr{F}_{t-1}]&=\1_{\neg C_i^{(t-1)}}\Pr[Y_S \geq x^{(t)}_S\ \forall S\ni e_i \mid Y_S \geq x^{(t-1)}_S\ \forall S\ni e_i ] \\&
    =\1_{\neg C_i^{(t-1)}}\prod_{S\ni e_i}\Pr\big[Y_S^i \notin [x_S^{(t-1)}, x_S^{(t)}]  \,\big\vert\, Y_S \geq x_S^{(t-1)}\big] \\
&= \1_{\neg C_i^{(t-1)}}\prod_{S\ni e_i}
        \bigg(\frac{1-x_S^{(t-1)}-\delta^{(t)}_S}{1-x_S^{(t-1)}}\bigg)^p
    \leq \1_{\neg C_i^{(t-1)}}\exp\left(-p\sum_{S\ni e_i}\delta^{(t)}_S\right).
  \end{align*}
The last inequality holds because
$(1-x) \leq \exp(-x)$.

For each $i=1, \dotsc, |X|$, we define a supermartingale $Z_i$ as follows:
we set $Z_i^{(0)} = 1$, $Z_i^{(i+1)} = 0$ and
$\smash{Z_i^{(t)}:=\1_{\neg C_i^{(t)}}\exp(p\sum_{S\ni e_i}x_S^{(t)})}$
for $t=1, \dotsc, i$.
First, we show that $Z_i$ is indeed a supermartingale. We observe that
\begin{align*}
\Exp[Z_i^{(t)}\mid \mathscr{F}_{t-1}]
    &=\Exp[\1_{\neg C_i^{(t)}} \mid \mathscr{F}_{t-1}] \cdot \exp\big(p\sum_{S\ni e_i}x_S^{(t)}\big)\\
    &\leq \1_{\neg C_i^{(t-1)}} \cdot \exp\big(-p\sum_{S\ni e_i}\delta^{(t)}_S\big)\exp\big(p\sum_{S\ni e_i}x_S^{(t)}\big)\\
    &= \1_{\neg C_i^{(t-1)}} \cdot \exp\big(p\sum_{S\ni e_i}x_S^{(t-1)}\big)\\
    &= Z_i^{(t-1)}.
\end{align*}
Now, consider a stopping time $\tau_i$ defined as follows:
  \begin{align*}
    \tau_i:=\begin{cases}
      i&\text{if } A_i\\
      i+1&\text{otherwise}.
    \end{cases}
  \end{align*}
Clearly, $\tau_i$ is bounded and by Proposition~\ref{prop:stopping}
we have  $\Exp[Z_i^{(\tau_i)}] \leq \Exp[Z_i^{(0)}]$. Therefore,
\begin{align*}
1 \geq \Exp[Z_i^{(\tau_i)}] = 
    \Exp\bigg[\1_{\neg C_i^{(i)}}\1_{A_i}\exp\big(p\sum_{S\ni e_i}x_S^{(i)}\big)\bigg].
  \end{align*}
  Whenever the event $A_i$ occurs, we must have $\sum_{S\ni e_i}x_S^{(i)}\geq 1$. Therefore the value of $\1_{A_i}\exp(p\sum_{S\ni e_i}x_S^{(t)})$ is either $0$ or greater than $\exp(p)$. This implies that
    $\Exp[\1_{\neg C_i^{(i)}}\1_{A_i}\exp(p)]\leq1$.
Therefore, the probability of $E_i$ can be bounded as
  \begin{align*}
    \Pr[E_i]&=\Exp[\1_{\neg C_i^{(i)}}\1_{A_i}]\leq \exp(-p).\qedhere
    \end{align*}
\end{proof}

\subsection{Online Non-Metric Facility Location}\label{subsec:online_nmfl}

To construct our algorithm for \nmfl, we use the following
linear programming relaxation.
\noindent\zsavepos{text-left-margin}\begin{alignat}{4}
  & \zsavepos{top-lp}\text{min} \quad \mathrlap{\displaystyle\sum\limits_{f\in F}^{} \fcost(f) x_f + \displaystyle\sum\limits_{(c,f)\in C\times F}^{}\ccost(c,f) x_{c,f}}\notag\\
  & \text{s.t.} \quad & \smashoperator{\sum_{f \in F}} x_{c,f}
    &\geq 1 & \quad & c \in C \label{LP:flow_atleast_one}\\
  & & x_f &\geq x_{c,f} && (c,f) \in C\times F \label{LP:facility_greater_edge}\\
  & & x_f, x_{c,f} &\in [0,1] && e\in E, f\in F\zsavepos{bottom-lp}\notag 
\end{alignat}

In the online setting, the clients $\{c_1, c_2, \ldots\}$ are revealed one-by-one. We denote by $\LP^{(i)}$ the linear program above with $C$ replaced by $C_{\leq i} := \{c_1, \ldots, c_i\}$ in constraint~(\ref{LP:flow_atleast_one}). By~\cite[Section~$3.2$]{alon_general_2006}, it is possible to maintain a monotonically increasing fractional solution $x^{(i)}$ to $\LP^{(i)}$, meaning 
\begin{displaymath}
    x^{(i)}_{c,f} \geq x^{(i-1)}_{c,f} \text{ and } x^{(i)}_f \geq x^{(i-1)}_f,
\end{displaymath}
that is $O(\log |F|)$-competitive with respect to the optimum value of $\LP^{(i)}$.

As in the set cover problem, we can assume that both the cost $\opt$ of the optimal offline solution and the value of $\log k$ are known to us in advance up to a factor of $2$. By scaling we can then assume that $k\leq \opt \leq k^2$. By rounding up every connection cost to its nearest power of two, we lose only a constant factor in the optimum value. We will assume that all connection costs smaller than $1$ are equal to $0$. As long as we only connect each client once, this leaves the new value of a solution within a constant factor of its old value. Consequently, we can assume that the connection costs are a power of $2$ and smaller than $k^2$.

The fractional solution to \nmfl is rounded online through a rounding algorithm for Set Cover, $\A_\setcov$. The online rounding algorithm can either be randomized, i.e.\,Algorithm~\ref{alg:rounding_SC_semiadaptive}, or deterministic, such as the procedure from~\cite{Alon_et_al_SetCover}. In the latter case, the corresponding set system must be known in advance to the online algorithm, while in the randomized setting the corresponding set system must be determined by the fixed \nmfl instance of the (semi-adaptive) adversary. We describe this construction in the proof of Claim~\ref{claim:semiadap_sc_applies}. The sets of this set system correspond to the facilities in $F$ and the element correspond to pairs $(c, t)$ with $c\in C$ and  $t\in \{0,2^0,2^1,\ldots 2^{\lceil 2\log k\rceil}\}$. Hence the set system has $O(|C|\cdot\log k)$ elements and $|F|$ sets. This yields a competitive ratio of $O(\log k \log|F|)$ for a randomized algorithm for \nmfl against a semi-adaptive adversary, and a $O(\log|C|\log|F|)$-competitive online algorithm in the deterministic setting. 

The algorithm works as follows. Upon arrival of client $c_i$, we compute $x^{(i)}$, the monotone solution to $\text{LP}^{(i)}$. 
We then compute the minimal $j_i \in \{0, 1, 2, \ldots, \lceil\log k\rceil\}$ ($j_i \in \{0, 1, 2, \ldots, \lceil\log|C|\rceil\}$ in the deterministic setting) such that 
\begin{alignat}{1}
    \sum_{\{f\in F \mid \ccost(f,c) \leq 2^{j_i}\}} x^{(i)}_{c_i, f} \geq \tfrac{1}{2}.\label{ineq:atleasthalf}
\end{alignat}
We then pass $(e_i, \{f\in F \mid \ccost(f,c) \leq 2^{j_i}\}, 2\cdot x_f^{(i)})$ to $\A_{SC}$. This outputs a set corresponding to facility $f$. We connect client $c_i$ to this facility $f$ and proceed to the next element. This procedure is summarized in Algorithm~\ref{alg:nmfl_to_sc}.
We note that the approach is conceptually similar to that used in \cite{bienkowski_nearly_2021}.
However, by making black-box use of an algorithm for online set cover we can more easily prove results for both the deterministic and the semi-adaptive setting.

\begin{algorithm}[H]
\begin{algorithmic}[1]

\State Initialize $\A_{\nmfl}$, an algorithm maintaining $x$, a fractional monotonically increasing solution to \nmfl.
\State Initialize $\mathcal{A}_{\setcov}$ as Algorithm~\ref{alg:rounding_SC_semiadaptive} or~\cite{Alon_et_al_SetCover}.\When {client $c_i$ arrives}
\State Pass client $c_i$ to $\A_{\nmfl}$. Update $x$.
\State Compute minimal $j_i\in \{0, 1, \ldots, \ell\}$ such that (\ref{ineq:atleasthalf}) holds.
\State Pass $(e_i, \{f \in F \mid \ccost(c_i, f) \leq 2^{j_i}\}, \{x^{(i)}\}_{f\in F})$ to $\A_{\setcov}$ and receive $f$. Connect $c_i$ to $f$.\label{line:SC_chooses_facility}
\EndWhen

\end{algorithmic}
\caption{Reducing Online Randomized NMFL to Online Randomized $\setcov$\label{alg:nmfl_to_sc}}
\end{algorithm}

\begin{proof}[Proof of Theorem~\ref{thm:nmfl}]
We show that Algorithm~\ref{alg:nmfl_to_sc} satisfies the desired
performance guarantees.

By the results of~\cite{alon_general_2006}, it is possible to deterministically (hence also in the semi-adaptive setting) maintain a $O(\log |F|)$-competitive, monotonically increasing solution $x$ to $\nmfl$. It remains to bound the expected cost of facilities and connection costs chosen through $\A_{\setcov}$ in Line~\ref{line:SC_chooses_facility}

\begin{claim}\label{claim:semiadap_sc_applies}
The (expected) cost of facilities opened in Line~\ref{line:SC_chooses_facility} of Algorithm~\ref{alg:nmfl_to_sc} is at most $O(\log|C| \log|F|)\cdot\Exp[\opt]$ (and $O(\log|C| \log|F|)\cdot \opt$ if $\A_\setcov$ is deterministic).\end{claim}

We proceed to bound the cost of the connection cost. For any client $c_i$, Algorithm~\ref{alg:nmfl_to_sc} connects it to some facility in the set $\{f\in F \mid \ccost(f,c) \leq 2^{j_i}\}$, incurring connection cost at most $2^{j_i}$. On the other hand, by Inequality~(\ref{ineq:atleasthalf}),
\begin{displaymath}
        \sum_{\{f\in F \mid \ccost(f,c) \geq 2^{j_i}\}} \ccost(e,f)\cdot x^{(i)}_{c_i, f} \geq \sum_{\{f\in F \mid \ccost(f,c) \leq 2^{j_i}\}} 2^{j_i}\cdot x_{c_i,f}^{(i)} \geq 2^{j_i-1}.
\end{displaymath}
Hence, summing over all clients $c_i$, two times the connection cost incurred by the fractional solution is an upper bound on the cost of the connection costs paid by Algorithm~\ref{alg:nmfl_to_sc}. Note that this inequality holds deterministically. Since the fractional solution is $O(\log |F|)$ competitive with respect to the optimum (in particular, also $\Exp[\opt]$), the claim follows.
\end{proof}

\begin{proof}[Proof of Claim~\ref{claim:semiadap_sc_applies}]

We begin by treating the case of a randomized online algorithm against a semi-adaptive adversary. To use Theorem~\ref{thm:rounding_fracSC}, we first need to show that the client-facility graph $\bar{G} := (\Bar{C}, F, \cost)$ implicitly corresponds to an instance $\bar{I}_{\setcov} := (\bar{X}, \S)$. Up to losing a constant factor and without loss of generality, we may assume that $\cost \in \{1, 2, 4, \ldots, 2^{\lceil \log(k)\rceil}\}$.\footnote{By a doubling strategy, the online algorithm can always be assumed to know the optimum value as well as $k$ the number of arriving elements. Any facility or connection with value larger than $\opt$ can then be ignored, and any facility or connection cost with value smaller than $\opt/k$ can be bought for total cost at most $2\cdot \opt$. This is analogous to the proof sketched in~\cref{subsec:guesskandopt}.} For each client $c\in C$ and each $j = 1, 2, \dotsc, \lceil\log k\rceil$,
we denote
\begin{alignat*}{1}
L(c,j) := \{f\in F \mid \ccost(c,f) \leq 2^j\},
\end{alignat*}
the set of facilities whose connection cost to $c$ is at most $2^j$.
From this, we build the corresponding set cover instance $\bar{I}_{\setcov} = (\bar{\mathcal{X}}, \mathcal{S})$ as follows. To each facility, we associate a corresponding set, $\mathcal{S} := \{S_{f_1}, \ldots, S_{f_{|F|}}\}$, each of same cost as the associated facility. In turn, each client $c_i\in C$ corresponds to $\lceil\log(k)\rceil$ different elements 
\begin{displaymath}
    e_0^{(i)}, e_1^{(i)}, \ldots, e_{\lceil\log(k)\rceil}^{(i)} \in \bar{\X}.
\end{displaymath}
We have that $e_j^{(i)} \in S_{f_k}$ if and only if $f_k \in L(c_i, j)$ - i.e.\,the connection cost of $c_i$ to $f_k$ is at most $2^j$. The ground set $\X$ has size $\lceil\log(k)\rceil\cdot|\Bar{C}| \leq O(|\Bar{C}|^2)$, and the set system $\S$ has size $|F|$. 

It remains to show that the fractional solution passed to $\A_\setcov$ is monotonically increasing and (fractionally) feasible. By Inequality~\ref{ineq:atleasthalf} and the inequality $x_f \geq \max_{c\in C} x_{c,f}$ implied by the LP constraint~(\ref{LP:facility_greater_edge}), the sets corresponding to $\{f \in F \mid \ccost(c_i, f) \leq 2^{j_i}\}$ receive fractional value at least $1/2$, hence, the fractional solution $2\cdot x_f$ passed to $\A_\setcov$ is feasible. Since the cost of the sets are the same as the cost of the sets, and the solution $x_f$ is $O(\log|F|)$-competitive, the fractional cost of $x$ is bounded by $O(\log|F|)\cdot \opt_\nmfl$. 

For the case of a deterministic algorithm with instance $(C, F)$, we can proceed analogously. The only difference is that we replace $\lceil \log k \rceil$ by $\lceil \log |C| \rceil$. The resulting set cover instance can be explicitly constructed in polynomial time.
\end{proof}

\newpage
\bibliographystyle{plainnat}
\bibliography{short_ref}

\newpage
\appendix

\section{Remarks}\label{sec:remarks}
\subsection{Assumptions on $k$ and $d_G(s,t)$}\label{subsec:guesskandopt}

In the descriptions and analysis of the algorithms we have assumed that $k$, the number of eventually arriving terminal pairs, is known, and that the $d_G(s,t)$ are polynomially related in $k$ (we have assumed $d_G(s,t) \in [1,k]$, though any polynomial relationship in $k$ is fine, up to constants). We now justify that this is without loss of generality.

\paragraph{Distances are polynomially related.}
First, we show how to remove the assumption on
the distances between terminals.
For each pair of terminals $(s_i,t_i)$ with penalty $p_i$,
we define $\alpha_i$ as
\begin{eqnarray}
    \alpha_i := \max_{j\leq i} \min(d_G(s_j, t_j), p_j).\label{def:alpha_i}
\end{eqnarray}
Note that $\alpha_i \leq \opt$.
We maintain a guess $\beta$ on $\opt$ as follows. Start with $\beta \gets \alpha_1$. Once $\alpha_i$ exceeds $k\cdot\beta$, we update $\beta \gets \alpha_i$ and rerun the algorithm from scratch.
This way, we always have $\beta \leq \opt \leq k^2 \cdot \beta$.
Between two updates of $\beta$, all demand pairs with $d_G(s,t) < \beta/k$ or with penalty $p<\beta/k$ are serviced greedily (with cost at most $\opt$). On the other hand, $k^2\cdot \beta$ is an upper bound on the current optimum value, in particular, there is no greedy path with cost greater than $k^2 \cdot \beta$. Hence, up to scaling the node-weights by $\beta$, we can assume the distance between
any pair of terminals is between $1$ and $k^3$.

\begin{algorithm}[H]
\begin{algorithmic}[1]
\State $\mathbf{Parameter: }$ value of $k$.
\State Initialize $\mathcal{A}$ an instance of Algorithm~\ref{alg:nwsf} that assumes that $d_G \in [1, k^3]$.
\State $\alpha_i \gets \min_{v\in V}w_v$, $\beta \gets \alpha_i$.
\When {$(s_i, t_i)$ arrives}
\State Update $\alpha_i$ as in~(\ref{def:alpha_i}).
\If{$\alpha_i > k\cdot\beta$}
\State Reininitialize a fresh instance of $\mathcal{A}$.
\State{Service all previously arrived terminals pairs $(s,t)$ with corresponding penalty $p$ greedily if $\min(d_G(s,t),p) \leq \beta/k$, and pass the remaining terminal pairs to $\mathcal{A}$.}
\EndIf
\State{Pass $(s_i, t_i)$ to $\mathcal{A}$, add the vertices it bought to $S$.}
\EndWhen

\end{algorithmic}
\caption{Distances are polynomially related.\label{alg:nwsf_without_opt}}
\end{algorithm}

\begin{lemma}
    \cref{alg:nwsf_without_opt} is $O(\log k \log n)$-competitive.
\end{lemma}

\begin{proof}
    Clearly, between two updates of $\beta$, $\alpha_i$ is a lower bound on the current optimum value $\opt_{\leq i}$. Since we only update $\beta$ once $\alpha_i > k\cdot \beta$, the optimum value between two updates is at most $k^2\cdot\alpha_i$. Given that we service all pairs with servicing cost at most $\beta/k$ greedily, the distances can be assumed to be between $1$ and $k^3$. 
    We now bound the cost incurred by Algorithm~\ref{alg:nwsf_without_opt}.
    Since the $\alpha_i$ are a lower bound on $\opt_{\leq i}$ (in particular, $\opt$), the cost incurred by greedy servicing costs can be upper bounded by
    \begin{displaymath}
        k\cdot(\alpha_1/k+\alpha_2/k + \ldots ) \leq k\cdot (1/k + 1/k^2 + \ldots)\cdot\opt.
    \end{displaymath}
    For the remaining costs incurred by rerunning the algorithm from scratch, note that $\opt_{\geq  i}$ doubles after every other update of $\beta$. Formally, denote by $\opt^{(j)}$ the value of the optimum right before we update $\beta$ for the $j^{\text{th}}$ time. Since the $\alpha_i$ are a lower bound on the respective optima, $\opt^{(j+2)} \geq 2\cdot \opt^{(j)}$. Setting $j_{\text{fin}}\in \N$ the final iteration, we can upper bound the (expected) cost by
    \begin{displaymath}
        \sum_{j_1 = 1}^{\lfloor j_{\text{fin}}/2\rfloor}(\ln k \ln n)\cdot \opt^{(2\cdot j)} + \sum_{j_1 = 1}^{\lfloor j_{\text{fin}}/2 \rfloor}(\ln k \ln n)\cdot \opt^{(2\cdot j + 1)} \leq 4\cdot (\ln k \ln n)\cdot \opt_\pcnwsf.
        \qedhere
    \end{displaymath}
\end{proof}

\paragraph{Knowledge of $k$.} We perform a standard doubling strategy on $\log k$: Starting with $k=2$, we run Algorithm~\ref{alg:nwsf} with $k$. Once the number of arrived terminals exceeds $k$, we replace $k$ by its square ($k\gets k^2$) and immediately rerun Algorithm~\ref{alg:nwsf} from scratch before proceeding to the current terminal pair. The formal description is given by Algorithm~\ref{alg:nwsf_without_k}.

\begin{algorithm}[H]
\begin{algorithmic}[1]
\State Initialize $\mathcal{A}$ an instance of Algorithm~\ref{alg:nwsf_without_opt} with $k = 2$.
\When {$(s_i, t_i)$ arrives}
\If{$i \in \{2^{2^1}, 2^{2^2}, \ldots \}$}
\State{Reinitialize $\mathcal{A}$ with $k \gets k^2$}
\State{Pass, one by one, $(s_1, t_1), \ldots, (s_{i-1}, t_{i-1})$ to $\mathcal{A}$ and add the bought vertices to $S$.}
\EndIf
\State{Pass $(s_i, t_i)$ to $\mathcal{A}$, add the vertices it bought to $S$.}
\EndWhen

\end{algorithmic}
\caption{Doubling Strategy for the online Node-Weighted Steiner forest Problem\label{alg:nwsf_without_k}}
\end{algorithm}

\begin{lemma}
    Algorithm~\ref{alg:nwsf_without_k} is $O(\log(k)\cdot\log(n))$-competitive.
\end{lemma}
\begin{proof}
    Denote by $\opt_{\leq 2^{2^j}}$ the optimum value of the online Steiner forest problem on the first $2^{2^j}$ demands $(s_1, t_1), \ldots, (s_{2^{2^j}}, t_{2^{2^j}})$. By assumption on the competitive ratio of Algorithm~\ref{alg:nwsf}, the total expenditure of Algorithm~\ref{alg:nwsf_without_k} is then upper bounded by
    \begin{alignat*}{1}       
    \sum_{i=1}^{\log\log k}\log(2^{2^i})\cdot\log(n)\cdot\opt_{\leq 2^{2^i}} &\leq \sum_{i=1}^{\log\log k}2^i\cdot \ln(n)\cdot \opt_{\leq 2^{2^i}}\\
        &\leq 4\cdot \log(k)\cdot\ln(n)\cdot \opt_{\leq k}.
    \end{alignat*}
    In the last inequality, we have used the fact that $\opt_{\leq i}$ only increases and the property of a geometric series. This concludes the proof.
\end{proof}

\subsection{Removing the assumption of knowing $|X|$}
\label{app:doubling_size_guess}
In \cref{alg:rounding_SC_semiadaptive} we assumed that the algorithm knows the size of the set $X'$ in advance. We used this information to choose $p:=2\log(k)$. We now show that this assumption can be removed by a doubling strategy.

As in \cref{lem:fixed-semi-adaptive-sc}, assume that $c^T x^{(i)}\leq B$ for all $i$. We start with $p=1$. Each time that an element is not covered by the rounded solution (i.e. the condition from line \ref{line:remains_uncovered}), we double $p$ and restart the algorithm, ignoring all the previously bought sets. So, in the $i$-th restart of the algorithm, we have $p=2^i$. Let $R_i$ be the event that the algorithm is restarted for the $i$-th time.
It follows from the proof of \cref{thm:semi-adaptive-sc} that the expected objective value in iteration $i$ is at most $O(2^i\cdot B)$. Hence, the total expected cost is at most:
\begin{align*}
  \sum_{i=0}^\infty \Pr[R_i]\cdot O(2^i\cdot B)
\end{align*}
Event $R_i$ only happens if one of the events $E_j$ happens in the run $i-1$. So, by \cref{claim:event_bound} and the union bound we have:
\begin{align*}
    \Pr[R_i]\leq |\X| \exp(-2^{i-1}).
\end{align*}
Hence:
\begin{align*}
  \sum_{i=0}^\infty \Pr[R_i]\cdot O(2^i\cdot B) &\leq \sum_{i=0}^{\log_2\log|\X'|+1} O(2^i\cdot B) + \sum_{i=\log_2\log|\X'| + 2}^{\infty} O(|\X|2^i\exp(-2^{i-1})\cdot B)
  \\&\leq O(\log|\X|\cdot B) +  \sum_{i=\log_2\log|\X|}^{\infty} O(\exp(\log(2)i-2^i)\cdot B)
  \\&\leq O(\log|\X|\cdot B)
\end{align*}
So, the total expected cost resulting from this algorithm is indeed $O(\log|\X|\cdot B)$.

\section{Flaw in previous work}\label{app:flaw}

The algorithm from~\cite{Markarian_2018} maintains an $O(\log n)$-competitive fractional solution to the linear relaxation that is monotonically increasing, which is rounded online.

Specifically, a \emph{single} variable $\mu\in [0,1]$ is chosen, that is the minimum of $2\lceil \log k\rceil$ uniform $[0,1]$ random variables.
Whenever the fractional solution $x_v$ exceeds $\mu$, node $v$ is purchased.

A key part of their analysis is proving that the event that a pair $(s, t)$ is not connected in the rounded solution occurs with probability at most $1/k^2$. This is done considering each vertex separator $K$ of $s$ and $t$. A fractional solution $x$ will satisfy $\sum_{v\in K}x_v\geq 1$. Since $\Pr[\mu \geq x_v]\leq(1-x_v)^{-2\log k}$, it is claimed that $\Pr[\mu \geq x_v\ \forall v\in K]\leq \prod_{v\in K}\Pr[\mu \geq x_v]\leq (1-x_v)^{-\sum_{v\in K}2\log k}$. However, the first inequality is clearly false, as the events $\mu \geq x_v$ are not independent.

Their lemma~\cite[Lemma~$2$]{Markarian_2018}, claims that with probability at least $(1-\tfrac{1}{k})$ over the random choice of $\mu$, the constructed solution is both $O(\log k)$-competitive (with respect to the online fractional solution) and integrally feasible (i.e. for each pair $(s_i, t_i)$, there is a path connecting $s_i$ to $t_i$).
This is wrong as can be seen from an instance of \nwst illustrated in
Figure~\ref{fig:counterexample}.
  With probability at least $(1-\frac{1}{2\log k})^{2\log k}\geq \exp(-2)$ we have $\mu \geq \frac{1}{2\log k}$. This means that for each $i$ with $\frac{1}{k-i+1}\leq \frac{1}{2\log k}$, the path connecting $s_1$ to $s_i$ is not bought through randomized rounding
  and the integral solution remains unfeasible.
  
  This cannot be fixed by buying greedy paths: the greedy path connecting $s_i$ and $s_1$ goes through $v_i$
  and a new greedy path is bought in every iteration.
Rewriting the previous inequality, we see that this happens for all $i\leq k-\log k$.
  Hence, the expected cost of the algorithm is at least $\exp(-4)(k-\log k)=\Omega(k)$, while the optimal cost is $O(1)$, resulting in a competitive ratio of $\Omega(k)$.

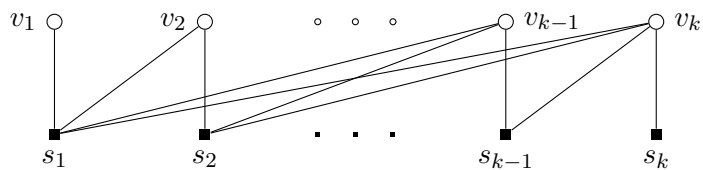
\begin{figure}[h]
  \centering
  \begin{tikzpicture}
      \node[label = -90:$s_1$, rectangle, inner sep = 2pt, fill] (s1) at (-4, 0) {};
      \coordinate[label = -90:$s_2$, rectangle, inner sep = 2pt, fill] (s2) at (-2, 0) {};
      \coordinate[rectangle, inner sep = 0.8pt, fill] (s3) at (-0.5, 0) {};
      \coordinate[rectangle, inner sep = 0.8pt, fill] (s4 at (0, 0) {};
      \coordinate[rectangle, inner sep = 0.8pt, fill] (s5) at (0.5, 0) {};
      \coordinate[label = -90:$s_{k-1}$, rectangle, inner sep = 2pt, fill] (s6) at (2, 0) {};
      \coordinate[label = -90:$s_k$, rectangle, inner sep = 2pt, fill] (s7) at (4, 0) {};

      \node[label = 180:$v_1$, circle,draw, inner sep = 2pt] (v1) at (-4, 1.5) {};
      \coordinate[label = 180:$v_2$, circle, draw, inner sep = 2pt] (v2) at (-2, 1.5) {};
      \coordinate[circle,draw, inner sep = 0.8pt] (v3) at (-0.5, 1.5) {};
      \coordinate[circle,draw, inner sep = 0.8pt] (v4) at (0, 1.5) {};
      \coordinate[circle,draw, inner sep = 0.8pt] (v5) at (0.5, 1.5) {};
      \coordinate[label = 0:$v_{k-1}$, circle,draw, inner sep = 2pt] (v6) at (2, 1.5) {};
      \coordinate[label = 0:$v_k$, circle,draw, inner sep = 2pt] (v7) at (4, 1.5) {};

      \draw (s1)--(v1);
      \draw (s1)--(v2);
      \draw (s1)--(v6);
      \draw (s1)--(v7);

      \draw (s2)--(v2);
      \draw (s2)--(v6);
      \draw (s2)--(v7);

      \draw (s6)--(v6);
      \draw (s6)--(v7);

      \draw (s7)--(v7);
  \end{tikzpicture}
  \caption[.]{Counterexample to~\cite{Markarian_2018}: The Steiner vertex $v_i$ is connected to $\{s_1,\ldots, s_i\}$, $\forall i\in [k]$. At time $i$,  $(s_1, s_i)$ arrives.
  Weight of each vertex $v_i$ is almost one up to a very small perturbation which causes that $w_{v_1} < w_{v_2} < \dotsb w_{v_k}$.
  Setting $x_{v_i}= \frac{1}{k-i+1}$ for all $i\in [k]$ gives a fractional solution that is feasible for all timesteps. In each step, we pass this fractional solution to the proposed online rounding algorithm.
  }
  \label{fig:counterexample}
\end{figure}

\end{document}